\documentclass[12pt]{article}
\usepackage{amssymb,amsmath}
\usepackage{amsthm}
\usepackage{graphicx}
\usepackage{caption}
\usepackage{subcaption}
\usepackage{color}
\usepackage{tikz-cd}

\makeatletter

\@addtoreset{equation}{section}
\def\section{\@startsection {section}{1}{\z@}{-2.5ex plus -1ex minus
 -.2ex}{1.3ex plus .2ex}{\large\bf}}
\def\subsection{\@startsection{subsection}{2}{\z@}{-2.25ex plus%
 -1ex minus -.2ex}{0.5ex plus .2ex}{\bf}}

\advance \voffset by -0.8in
\advance \hoffset by -0.8in
\def\dd{\mathrm{d}}

\def\pie{\pi^{\text{\tiny eq}}}
\def\pic{\pi^{\text{\tiny co}}}
\def\Pie{\Pi^{\text{\tiny eq}}}
\def\Pic{\Pi^{\text{\tiny co}}}
\def\bee{\begin{equation}}
\def\eee{\end{equation}}
\def\bea{\begin{align}}
\def\eea{\end{align}}

\def\Ad{\mbox{Ad}}

\def\bp{{\mbox{\boldmath $p$}}}

\def\bq{{\mbox{\boldmath $q$}}}
\def\be{{\mbox{\boldmath $e$}}}
\def\bk{{\mbox{\boldmath $k$}}}
\def\bx{{\mbox{\boldmath $x$}}}
\def\by{{\mbox{\boldmath $y$}}}
\def\bq{{\mbox{\boldmath $q$}}}

\def\bs{{\mbox{\boldmath $s$}}}

\def\bpm{\begin{pmatrix}}
\def\epm{\end{pmatrix}}

\newcommand{\NN}{\mathbb{N}}
\newcommand{\ZZ}{\mathbb{Z}}

\newcommand{\RR}{\mathbb{R}}
\newcommand{\CC}{\mathbb{C}}
\newcommand{\tr}{{\rm tr}}
\def\bee{\begin{equation}}
\def\eee{\end{equation}}
\newtheorem{theorem}{Theorem}[section]
\newtheorem{lemma}[theorem]{Lemma}

\newtheorem{proposition}[theorem]{Proposition}

\newtheorem{definition}[theorem]{Definition}

\def\Lor{L_3^{+\uparrow}} 

\providecommand{\abs}[1]{\lvert#1\rvert}

\usepackage{nicefrac}
\usepackage{braket}
\usepackage{comment}

\begin{document}
\parskip 3pt
\parindent 8pt

\begin{center}

{\Large \bf Non-commutative waves  for gravitational anyons }

\baselineskip 20 pt

\vspace{.2cm}

{ \bf Sergio Inglima and  Bernd~J.~Schroers}   \\
Department of Mathematics and Maxwell Institute for Mathematical Sciences \\
 Heriot-Watt University, 
Edinburgh EH14 4AS, United Kingdom \\ 
\tt{sfi1@hw.ac.uk} and \tt{b.j.schroers@hw.ac.uk}

\vspace{0.3cm}

{December 2018}
\baselineskip 16 pt

\end{center}

\begin{abstract}
\noindent  
We revisit the representation theory  of the quantum double of the universal cover of the  Lorentz group in 2+1 dimensions,  motivated by its role as  a deformed Poincar\'e symmetry and symmetry algebra in (2+1)-dimensional quantum gravity. We express the unitary irreducible representations in terms of covariant, infinite-component fields on curved  momentum space satisfying algebraic spin and mass constraints. Adapting and applying the method of  group Fourier transforms, we obtain covariant fields on (2+1)-dimensional Minkowski space which necessarily depend on an additional internal and  circular dimension. The momentum space constraints turn into differential or exponentiated differential operators, and the  group Fourier transform induces a  star product on Minkowski space and the internal space which is essentially a version of Rieffel's deformation quantisation via convolution.
\end{abstract}


\section{Introduction}

The possibility of anyonic statistics in  two spatial dimensions lies at the root of  the peculiarity and intricacy of  planar phenomena in  quantum physics, ranging from the quantum Hall effect to potential uses of anyons in topological quantum computing \cite{AnyonBook}. The mathematical  origin of anyonic  statistics is the infinite connectedness of the planar rotation group $\mathrm{SO}(2)$ which is topologically a circle. The goal of this paper, in brief, is to explore the consequences of this fact for quantum gravity in 2+1 dimensions, where  the infinite connectedness  is doubled  and appears in both  momentum space  and  the rotation group.

Our strategy  in  pursuing this goal is to extend and generalise the method developed in \cite{SemiDual,SchroersWilhelm}, which proceeds from a definition of the   symmetry algebra, via a covariant formulation of its irreducible  unitary representations (UIRs) in momentum space to, finally,  covariant fields in spacetime obeying  differential or finite-difference  equations.  One of the upshots of our work  is  a construction of non-commutative plane waves  for  gravitational anyons via a group Fourier transform. While this transform  arose in relatively recent literature in quantum gravity \cite{FL,FM,Raasakka,GOR} we note that it is essentially  an example and extension of Rieffel's deformation quantisation of the canonical Poisson structure on the  dual of a Lie algebra via convolution \cite{Rieffel}.

Anyonic behaviour occurs in both non-relativistic and relativistic physics, and for the same topological reason. The  proper and  orthochronous Lorentz group $\Lor$  in 2+1 dimensions retracts to $\mathrm{SO}(2)$, and is  therefore infinitely connected. Its universal cover, which  cannot be realised as a matrix group, governs the properties of relativistic anyons.  As we shall review,  the double cover of $\Lor$ is the matrix group $\mathrm{SU}(1,1)$ or, equivalently, $ \mathrm{SL}(2,\RR)$. In this paper, we  write the universal cover  as $\widetilde{\mathrm{SU}}(1,1)$.

In the spirit of  Wigner's classification of particles \cite{Wigner},  
relativistic anyons are  classified by the UIRs of the universal cover of the (proper, orthochronous)  Poincar\'e group in 2+1 dimensions, which is the semidirect product of $\widetilde{\mathrm{SU}}(1,1)$ with the group of spacetime translations. It is natural to identify the latter with the dual of the  Lie algebra of $\mathrm{SU}(1,1)$. Then the universal cover of the Poincar\'e group is $\widetilde{\mathrm{SU}}(1,1)\ltimes \mathfrak{su}(1,1)^*$, with the first factor acting on the second by the co-adjoint action. 
While the UIRs of the Poincar\'e group in 2+1 dimensions are labelled by a real mass parameter  and an integer spin parameter, the corresponding mass and spin parameters for the universal cover both take arbitrary real values \cite{Grigore}. In  other words, going to the universal cover puts mass and spin on a more equal footing.

The canonical construction of the UIRs of the Poincar\'e group in 2+1 dimensions leads to states being realised as functions on momentum space, obeying constraints \cite{Binegar}. Writing these constraints in a Lorentz-covariant way, and Fourier transforming leads to  standard  wave equations of relativistic physics like the Klein-Gordon or Dirac equation. Relativistic wave equations for anyons have also been constructed, but for spins which are not half-integers they require infinite-component wave functions, and the derivation of the equations  is not   straightforward \cite{Gitman,JackiwNair,Plyushchay1,Plyushchay2}. 

Our treatment will naturally lead us to an equation derived via a different route      by Plyushchay in \cite{Plyushchay1,Plyushchay2}. It makes use of the discrete series UIR of $\widetilde{\mathrm{SU}}(1,1)$, but, as pointed out by Plyushchay, it is essentially a dimensional reduction of an equation already studied by Majorana \cite{Majorana,ST}.

The focus of this paper is  a deformation of the universal cover of the Poincar\'e group to a quantum group which arises in (2+1)-dimensional quantum gravity. In 2+1 dimensions, there are no propagating gravitational degrees of freedom, and the phase space of gravity interacting with a finite number of particles and in a universe where spatial slices  are  either compact or accompanied by suitable boundary conditions at spatial infinity  is finite-dimensional. In those cases where the resulting phase space could be quantised, the Hilbert space of the quantum theory can naturally be constructed out of unitary representations of the   quantum double of $\Lor$  or one of its covers \cite{MeusburgerSchroers1,CMBSCombQuant,LQGQDouble}. In a sense that can be made precise, this  double is a deformation of the Poincar\'e group in  2+1 dimensions \cite{EuclidSchroers,BaisScatt}, with the   linear momenta  in $ \mathfrak{su}(1,1)^*$ (the generators of  translations)  being replaced by functions on $\Lor$  or one of its covers.

In the quantum  double of $\Lor$,  Lorentz transformations and translation are implemented via  Hopf algebras which are in duality, namely   the group algebra of $\Lor$  and the dual algebra of functions on $\Lor$. The quantum double is a ribbon Hopf algebra whose $R$-matrix can be given explicitly. The  unitary irreducible  representations  describing massive particles are labelled by an integer spin and a mass parameter taking values on a circle.

In analogy with the  treatment of the Poincar\'e group, one could define the  universal cover of the  double of $\Lor$  by replacing the group algebra of the Lorentz group with  the group algebra of the universal cover $\widetilde{\mathrm{SU}}(1,1)$. There are  indications that this is required  when applying the quantum double to quantum gravity in 2+1 dimensions. In particular,  several independent arguments lead to the conclusion that, in (2+1)-dimensional gravity, the spin $s$ is quantised  in units which depends on its mass $m$ according to 
\begin{equation}
s=\frac{n}{1-\frac{\mu}{2\pi}},
\label{SpinQuant}
\end{equation}
where $n \in \ZZ$, $\mu =8\pi mG$ and  $G$ is Newton's gravitational constant, see   \cite{BaisScatt,Matschull}.

Covering the Lorentz transformations without changing the momentum algebra would destroy the duality between the two.  
It is therefore more natural to also consider a universal covering of the momentum algebra, i.e.,  to identify the momentum  algebra with the function algebra on  $\widetilde{\mathrm{SU}}(1,1)$. The resulting quantum double of the universal cover of the Lorentz group, called Lorentz double in \cite{BaisScatt},  is a  ribbon-Hopf algebra and has UIRs   describing massive particles  which are  labelled by a  spin parameter  $s$  and a mass  parameter $\mu$ for which \eqref{SpinQuant} makes sense.  

In this paper, we consider  the Lorentz double and derive a new formulation of  its  UIRs in terms of infinite-component functions on momentum space $\widetilde{\mathrm{SU}}(1,1)$ obeying  Lorentz-covariant constraints. We  extend and use the notion of group Fourier transforms \cite{FL,FM,Raasakka,GOR}  to derive covariant wave equations  on Minkowski space equipped with a $\star$-product.

  Our method extends  earlier work  in \cite{SchroersWilhelm} where  analogous relativistic wave equations for massive particles were obtained from the UIRs of the quantum double of the Lorentz group. The transition to the universal cover poses two separate challenges. Our  wave functions  in momentum space now live on $\widetilde{\mathrm{SU}}(1,1)$,  and  they take values in   infinite-dimensional UIRs of $\widetilde{\mathrm{SU}}(1,1)$ called  the discrete series.  The change in momentum space  leads to a particularly natural version of  the  group Fourier transform, essentially because group-valued momenta can be parametrised  bijectively via the exponential map and one additional  integer label. In order to include the integer  in our Fourier transform we are forced to introduce a dual circle on the spacetime side. The emergence of a compact additional dimension is a remarkable and intriguing aspect of our construction. 
  
  Fourier transforming the  algebraic spin and mass constraints from momentum space to Minkowski space produces  equations on non-commutative Minkowski space which involve either differential operators or    the exponential  of a first-order differential operator.  We end the paper with a short discussion of these non-commutative wave equations for gravitational anyons, leaving a detailed study for future work. Non-commutative waves for anyons have previously been discussed in the literature \cite{HP}. However, the discussion there is in  the  context  of non-relativistic limits    rather than the inclusion of gravity, and the mathematics is rather different.   
  
The paper is organised as follows.
 In Sect.~2, we introduce our notation and  review the definition and parametrisation  of the universal cover  $\widetilde{\mathrm{SU}}(1,1)$ as well as the discrete series UIRs. We also revisit the UIRs of the Poincar\'e group and  briefly summarise the covariantisation of the UIRs and their Fourier transform, following \cite{SchroersWilhelm}. In Sect.~3, we generalise  the covariantisation procedure to the universal cover of the Poincar\'e group, thus obtaining wave equations  for infinite-component anyonic wave functions. Our version of these equations is essentially that considered by Plyushchay \cite{Plyushchay1,Plyushchay2} but  our derivation of them appears to be   new.  Sect.~4 extends the analysis to the Lorentz double.  We derive a Lorentz-covariant form of the UIRs, and point out that one of the defining constraints, called the spin constraint, can be expressed succinctly in terms of the ribbon element of the  Lorentz double. The group Fourier transform on 
 $\widetilde{\mathrm{SU}}(1,1)$ requires a parametrisation of this group via the exponential map, and we discuss this in some detail. 
 We use the Fourier transform to derive non-commutative wave equations, and then use  the  short final Sect.~5  to discuss our results and to point out avenues for further research.

\section{Poincar\'e symmetry  and massive particles in 2+1 dimensions}
\label{conventions}
We review the symmetry  group  of  (2+1)-dimensional  Minkowski space:  its Lie algebra, its various covering  groups and their representation theory.  Unfortunately there is no single book or paper which covers these topics in conventions which are convenient for our purposes.  We therefore adopt a mixture of the conventions in the papers \cite{BaisScatt} and  \cite{SchroersWilhelm}  and the book \cite{Sally}, which are key references for us. 

\subsection{Minkowski space and the double cover of the Poincar\'e group}
\label{poincareconventions}
We denote $(2+1)$ dimensional Minkowski space by $\mathbb{R}^{2,1}$, with the convention that the Minkowski metric is mostly minus. Vectors in  $\mathbb{R}^{2,1}$ will be denoted  by $\bx=(x^0,x^1,x^2)$,  with Latin indices for components and the inner product given by
\begin{equation}
\eta(\bx,\by)=\eta_{ab}x^ay^b=x^0y^0-x^1y^1-x^2y^2.
\end{equation}

The group of linear transformations that leave $\eta$ invariant is the Lorentz group $L_3=O(2,1)$. This group has four connected components, but we  are mainly interested in the component connected to the identity, i.e., the subgroup of proper orthochronous Lorentz transformations, denoted $\Lor$.

The group of affine transformations that leave  the Minkowski metric invariant is the semi-direct product $L_3\ltimes \RR^3$ of the Lorentz group with the abelian group of translations. We call its identity component  Poincar\'{e} group  and denote it as 
\begin{equation}
P_3=\Lor \ltimes \RR ^3.
\end{equation}

The Lie algebra  $\mathfrak{p}_3$ of the Poincar\'e group is spanned by rotation and boost generators $J^0,J^1,J^2$, and time and space translation generators $P_0,P_1,P_2$. They can be chosen so that the  commutators are 
\bee
\label{poincom}
[J^a,J^b]=\epsilon^{abc}J_c, \quad [J_a,P_b]=\epsilon_{abc} P^c, \quad [P_a,P_b]=0.
\eee
Here and in the following $\epsilon_{abc}$ is the totally antisymmetric tensor with $\epsilon_{012}=1$, and indices are raised and lowered with the Minkowski metric $\eta=\text{diag} (1,-1,-1)$. Note that this means in particular that $-J^0$ is the generator of mathematically positive rotations in momentum space. This affects our conventions for the sign of  the spin later in this paper.

In quantum mechanics, classical symmetries described by a Lie group $G$ are implemented by projective representations, which, in the case of the Lorentz group,  may equivalently  be described by unitary representations of the universal covering group. In the case of relativistic symmetries in 3+1 dimensions the universal cover of the Lorentz group is its double cover and is isomorphic to $\mathrm{SL(2,\CC)}$. In 2+1 dimensions the double cover of $\Lor$ is isomorphic to $\mathrm{SL}(2,\RR)$ or, equivalently,  $\mathrm{SU}(1,1)$,  but this is not the universal cover. Therefore a choice has to be made as to which covering group one should implement in the quantum theory. The paper  \cite{SchroersWilhelm},  to which we will refer frequently for details, works with the  double cover in the realisation $\mathrm{SL}(2,\RR)$.

 We will consider the universal cover here. For our purposes it is  more convenient to  first consider the double cover  $ \mathrm{SU}(1,1)$ because it allows for an easy transition to the universal cover.  In order to translate   results  from \cite{SchroersWilhelm},  the reader will need to apply  the unitary matrix
\bee
\label{htraf}
h=\frac 12 \bpm 1+i & -1-i \\ 1-i & \phantom{-}1-i \epm,
\eee
which conjugates  $\mathrm{SL}(2,\RR)$ into $ \mathrm{SU}(1,1)$  within $\mathrm{SL}(2,\CC)$, i.e., $h^{-1}\mathrm{SL}(2,\RR)h= \mathrm{SU}(1,1)$.

 We use the following basis of the Lie algebra $\mathfrak{su}(1,1)$
\bee
\label{su(1,1) algebra}
s^0 =-\frac{i}{2}\sigma_3=\bpm -\frac{i}{2} & 0 \\ \phantom{-} 0 & \frac{i}{2} \epm, \quad 
s^1 =-\frac{1}{2}\sigma_1=\bpm \phantom{-} 0 & -\frac{1}{2} \\ -\frac{1}{2}  & \phantom{-} 0 \epm, \quad 
s^2 =\frac 12 \sigma_2=\bpm 0 & -\frac{i}{2} \\ \frac{i}{2} & \phantom{-} 0 \epm,
\eee
which are normalised to have the commutation relations  of the Lorentz part of \eqref{poincom}: 
\bee
\label{su11coms}
[s_a,s_b]=\epsilon_{abc}s^c.
\eee
The basis  $t^a$, $a=0,1,2,$ of   $\mathfrak{sl}(2,\RR)$  used in  \cite{SchroersWilhelm} 
is related to  our basis of $\mathfrak{su}(1,1)$  via $
s^a=h^{-1}t^ah.$

An arbitrary matrix $M \in \mathrm{SU}(1,1)$  can be parametrised in terms of two complex numbers $a,b$ which satisfy $|a|^2 -|b|^2 =1$ via
\bee
\label{Mab}
 M = \bpm a & \bar{b} \\ b & \bar{a} \epm . 
\eee
However,  a more convenient parametrisation for the  extension to the universal cover is obtained
by introducing an angular coordinate $ \omega \in [0,4\pi)$ and a complex number $\gamma=b/a$ of modulus $|\gamma|<1$
so that 
\bee
\label{abangles}
a=\frac{1}{\sqrt{1-|\gamma|^2}}e^{i \frac{\omega} {2}}  , \qquad 
b=\frac{\gamma}{ \sqrt{1-|\gamma|^2} }e^{ i\frac{\omega} {2} }.
\eee
Using this parametrisation, one can see that $\mathrm{SU}(1,1)$ is topologically the open solid torus  $S^1 \times D$, where $D$ is the open  unit disk. 
Note that an   $\mathrm{SU}(1,1)$ element parametrised via \eqref{abangles}   can be written as 
\bee
\label{Momegamma}
M(\omega,  \gamma)  =\frac{1}{\sqrt{1-|\gamma|^2}}\bpm  e^{i\frac{\omega }{2}  } & \bar{\gamma} e^{-i \frac{\omega }{2}}   \\
\gamma e^{i \frac{\omega }{2}}  & e^{-i \frac{\omega }{2}}  \epm= \frac{1}{\sqrt{1-|\gamma|^2}}
\bpm 1 & \bar{\gamma}   \\
\gamma  &   1 \epm \bpm  e^{i\frac{\omega }{2}  } & 0   \\
0 & e^{-i \frac{\omega }{2}}  \epm.
\eee

We will need to understand the conjugacy classes of $\mathrm{SU}(1,1)$ for various applications in this paper. These are obtained from the well-known conjugacy classes of $\mathrm{SL}(2,\RR)$ (see \cite{SchroersWilhelm}, for example) by conjugation with $h$ and therefore determined by the trace (which, for determinant one, fixes the eigenvalues up to ordering) plus additional data. 
Elements $u\in \mathrm{SU}(1,1)$ with absolute value of the  trace less than 2  are called elliptic elements. They are conjugate to rotations  of the form
\bee
\label{elliptic}
\begin{pmatrix} e^{i\frac{\alpha}{2}}  & 0  \\
0 & e^{-i\frac{\alpha}{2}} 
\end{pmatrix},  \quad  \alpha \in (0,2\pi) \cup (2\pi, 4\pi).
\eee
Each value of $\alpha$ labels one conjugacy class (so there are two conjugacy classes for each value of the trace). 
Elements with absolute value of the trace equal to 2  include  $\pm \text{id}$ (forming a conjugacy class each) and parabolic elements, which are conjugate to
\bee
\label{parabolic}
 \begin{pmatrix} \pm1-i\zeta & \zeta  \\
\zeta &  \pm 1+i\zeta
\end{pmatrix}, \quad \zeta \in (-\infty, 0)\cup (0,\infty).
\eee
Each choice of sign on the diagonal and each value of $\zeta$ labels one conjugacy class.
Element with  absolute value of the trace greater than two are called hyperbolic. They are conjugate to
\bee
\label{hyperbolic}
\pm  \begin{pmatrix} \cosh \xi & \sinh\xi  \\
\sinh\xi &  \cosh \xi
\end{pmatrix}, \quad \xi \in (0,\infty).
\eee
 Each choice of sign and each  value of $\xi$ determines one conjugacy class (negative values are not needed since  conjugation with $\sigma_3\in \mathrm{SU}(1,1)$ flips the sign).

The  double cover of the Poincar\'e group is
\begin{equation}
\tilde{P}_3 \simeq \mathrm{SU}(1,1) \ltimes \RR^3,
\end{equation}
but  for our purposes  it is natural  to identify the translation subgroup with the vector space $\mathfrak{su}(1,1)^*$, 
and  to view $\tilde{P}_3$ as 
\bee
\label{P3Iso}
\tilde{P}_3 \simeq \mathrm{SU(1,1)} \ltimes \mathfrak{su}(1,1)^*, 
\eee
with $\mathrm{SU(1,1)}$ acting  on translations via the co-adjoint action $\Ad^*$.  More precisely, the product of elements $(g_1,a_1), (g_2,a_2) \in \tilde{P}_3$ in  the conventions of  \cite{SchroersWilhelm,SemiDual} is 
\bee
\label{product}
(g_1,a_1) (g_2,a_2) = (g_1g_2,\Ad^*_{g_2}a_1 +a_2).
\eee 

The motivation for the interpretation \eqref{P3Iso} of the double cover of the Poincar\'e group  comes from 
the formulation of 3d gravity as a Chern-Simons theory with the Poincar\'e group as a gauge group. This  requires an invariant and non-degenerate pairing on $\mathfrak{p}_3$.  The, up  to scale, a unique such pairing is the dual pairing between the  Lorentz generators and the translation generators, i.e., the canonical pairing on the Lie algebra 
\bee
\label{p3algebra}
\mathfrak{p}_3=  \mathfrak{su}(1,1)\oplus  \mathfrak{su}(1,1)^*.
\eee

As remarked in \cite{SchroersWilhelm} and for the Euclidean case in \cite{SemiDual},  the  dual  basis of  the translation generators $P_a$ (now viewed as a basis of $ \mathfrak{su}(1,1)^*$)  may have a different normalisation to the basis ${s^a}$  which is fixed by the commutation relations \eqref{su(1,1) algebra} .
 Adapting the conventions of \cite{SemiDual,SchroersWilhelm}, and working with the generators $s^a$ for $ \mathfrak{su}(1,1)$, we 
   write the pairing in our basis  as 
\bee
\label{pairing}
\langle s_a, P_b\rangle = -\frac 1 \lambda \eta_{ab},
\eee
where  the scale  $\lambda $ is a constant of dimension inverse mass. In the context of 2+1 gravity it is  related to Newton's gravitational constant $G$  via $\lambda=8\pi G$. 

With the usual identification of  momentum space as the dual of spacetime translations, momenta $ p$  in 2+1 dimensions now naturally live in   $(\mathfrak{su}(1,1)^*)^*\simeq \mathfrak{su}(1,1)$.
  For consistency with \eqref{pairing} we set
\begin{equation}
P^{*a}=-\lambda s^a,
\end{equation}
so that  a general element  $p$ of momentum space may be expressed as 
\begin{align}
\label{pexpand}
p=p_a P^{*a}=- \lambda p_as^a.
\end{align}
The reason for inserting a minus sign here  is that  $- s^0$ generates  mathematically positive rotation, as mentioned after \eqref{poincom}.

While the pairing between Lorentz and translation generators is canonical (up to scale), the  inner product between two momenta $p$ and $q$ in $\mathfrak{su}(1,1)$  requires the trace. We define the inner product as  
\begin{equation}
-\frac{2}{\lambda ^2}\text{tr}(pq)=p_aq^a.
\end{equation}
This is invariant under the (adjoint) action of $\mathrm{SU(1,1)}$.  Using the  convention to mark coordinate vectors by bold letters we 
write
\bee
\bp\cdot \bq = p_aq^a, 
\eee  
and occasionally also 
\bee
\bp\cdot \bs = p_as^a.
\eee
We  also require a notation for the absolute value of the norm of $\bp$. We  write this as
\bee
|\bp| = \sqrt{ |\bp\cdot \bp|}.
\eee

As pointed out in the Introduction,  $\mathrm{SU(1,1)}$ and its universal cover both   play two  mutually  dual roles in this paper as Lorentz symmetry and as curved momentum space. In the latter context we will require one further parametrisation of  ${\mathrm{SU}}(1,1)$, namely the one obtained  via the exponential map.  By `conjugating' the detailed discussion of the corresponding parametrisation of $\mathrm{SL}(2,\RR)$ in  \cite{SchroersWilhelm} with \eqref{htraf},  one checks that  any element $u\in {\mathrm{SU}}(1,1)$ can, up to  a sign, be written as the exponential of the Lie-algebra valued momentum $p$, i.e., 
\bee
\label{exponential}
u =\pm \exp(p).
\eee
We will revisit this fact and its geometrical interpretation in Sect.~\ref{anywave}. In that context we will also need the following explicit form of the exponential: 
\bee
\label{expoformula}
\exp(-\lambda p_as^a)= c(|\bp|)\,\text{id} - 2 s(|\bp|)\hat{p}\cdot\bs,
\eee
where we introduced the function 
\bee
s(|\bp|) = \begin{cases} \sin\left(\lambda \frac{ |\bp| }{2}\right)  \; & \text{if} \;  \bp^2  >0 
\\ \lambda \; & \text{if} \; \bp^2  =0 \\ \sinh\left(\lambda\frac{|\bp| }{2}\right)\; & \text{if} \;  \bp^2  <0 
 \end{cases} \quad \text{and} \quad 
 c(|\bp|) = \begin{cases} \cos\left(\lambda \frac{|\bp| }{2}\right)  \; & \text{if} \;  \bp^2  >0 
\\ 1  \; & \text{if} \;  \bp^2  =0 \\ \cosh\left(\lambda\frac{|\bp|}{2}\right)\;&  \text{if} \;  \bp^2  <0 
 \end{cases}, 
\eee
as well as the generalised unit-vector
\bee
\hat{p} = \begin{cases}\frac{\bp}{|\bp|}\; & \text{if} \; \bp^2 \neq  0 
\\ \bp \; & \text{if} \; \bp^2 =0  
 \end{cases}.
\eee
The functions $s$ and $c$ satisfy a generalised version of the Pythagorean formula:
\bee
\label{pythagoras}
c^2(|\bp|) + \hat{p}^2 s^2(|\bp|)=1.
\eee

\subsection{The universal cover the Lorentz group  and the discrete series representation } \label{SL2 Coverings}
In this paper we shall consider anyonic particles and for this we will need to work with the universal cover of the Poincar\'{e} group which we denote  by  $P^\infty_3$: 
\begin{equation}
\label{UCovPoinc}
P^\infty_3 =\widetilde{\mathrm{SU}}(1,1) \ltimes \mathfrak{su}(1,1)^*,
\end{equation}
 where $\widetilde{\mathrm{SU}}(1,1)$ is the universal cover of $\mathrm{SU}(1,1)$. The group  $\widetilde{\mathrm{SU}}(1,1)$ is not a matrix group, but, as we will explain in some detail below, its elements are conveniently parametrised by a real number $\omega$ and complex number $\gamma$ of modulus less than $1$. Thus we  write elements  of $\widetilde{\mathrm{SU}}(1,1)$ as pairs $(\omega,\gamma)$ and elements of $P^\infty_3$
  as pairs $((\omega, \gamma),a)$, with   $a\in\mathfrak{su}(1,1)^*$, so that the product is 
 \bee
\label{productcover}
((\omega_1,\gamma_1),a_1) ((\omega_2,\gamma_2),a_2) = ((\omega_1,\gamma_1)(\omega_2,\gamma_2),\Ad^*_{(\omega_2,\gamma_2)}a_1 +a_2).
\eee

The universal cover of $\mathrm{SU}(1,1)\simeq S^1\times D$,  denoted $\widetilde{\mathrm{SU}}(1,1)$,  is diffeomorphic to  the product $\RR \times D$, where  the circle factor $S^1$ has been covered by the real line. We extend our parametrisation \eqref{Momegamma} of $\mathrm{SU}(1,1)$ to the universal cover by simply allowing the angular variable $\omega$ to take values in $\RR$. We then identify elements of  $\widetilde{\mathrm{SU}}(1,1)$   with pairs $(\omega, \gamma) \in \RR \times D$.  The product  $(\omega_1,\gamma_1)(\omega_2,\gamma_2)$  is the element $(\omega,\gamma)$ given by an analytic extension of the  formulae one obtains when writing the matrix product  of elements of the form \eqref{Momegamma} in terms of the parameters\footnote{The group product is essentially the one given in \cite{Sally} except that we have parametrised the $S^1$ in $\mathrm{SU}(1,1)$ with $[0,4\pi)$ rather than $[0,2\pi)$. This  ensures that a final projection to the Lorentz group results in a complete spatial rotation having an angle of $2\pi$.}:
\begin{align}
\label{SU(1,1)tildeproduct}
\gamma &= (\gamma_1+ \gamma_2e^{-i\omega_1})(1+ \bar{\gamma}_1 \gamma_2e^{-i\omega_1})^{-1} \nonumber \\
\omega &= \omega_1 +  \omega_2  + \frac 1 i \ln\left(\frac{1+ \bar{\gamma}_1 \gamma_2e^{-i \omega_1} }{ 1+ \gamma_1 \bar{\gamma_2}e^{i\omega_1}}\right).
\end{align}
Here the logarithms of the form $\ln (1+w)$ are defined in terms of the usual (Mercator) power series, which returns the principal value of  $\ln(1+w)$. With $\gamma_1$ and $\gamma_2$ inside the unit disk, this means  in particular
\bee
\frac 1 i \ln\left(\frac{1+ \bar{\gamma}_1 \gamma_2e^{-i \omega_1} }{ 1+ \gamma_1 \bar{\gamma_2}e^{i\omega_1}}\right) \in \left(- \pi , \pi \right).
\eee
We also note that 
\bee
(\omega,\gamma)^{-1}=(-\omega, -e^{i\omega} \gamma).  
\eee

In terms of the coordinates $(\omega,\gamma)$, the canonical projection 
\bee
\label{canproj}
\pi : \widetilde{\mathrm{SU}}(1,1) \rightarrow \mathrm{SU}(1,1)
\eee
is  simply the map
\bee
\label{coverproject}
\pi: (\omega, \gamma) \mapsto  M( \omega, \gamma),
\eee
where it is clear from \eqref{Momegamma} that the right hand side only depends on  $\omega$ mod $4\pi$.  This is a homomorphism whose kernel is the central subgroup generated by the element $(4\pi,0)$. It is, however, convenient to introduce a special name for rotations by $2\pi$. With 
\bee
\label{Omegadef} 
\Omega=(2\pi,0) \in \widetilde{\mathrm{SU}}(1,1),
\eee
we have $\pi(\Omega) = -\text{id}$ and 
\bee
\label{Omegaker}
\text{ker} \;\pi =\{\Omega^{2n} | n\in \ZZ\}.
\eee

When working with the Lie algebra  of $\widetilde{\mathrm{SU}}(1,1) $ we will  continue to use the notation  $s^a$, $a=0,1,2$, for the generators, i.e., we identify the  Lie algebra of $\widetilde{\mathrm{SU}}(1,1) $  with its  image  under the differential of the projection \eqref{canproj}.  However,  the exponential of  these generators  $\widetilde{\mathrm{SU}}(1,1) $ cannot be computed via matrix exponentials; instead one has to use the differential geometric definition in terms of the geodesic flow with respect to the bi-invariant (Lorentzian) metric on $\widetilde{\mathrm{SU}}(1,1) $.  We will need to exponentiate timelike, lightlike and spacelike generators at various points  in this paper, and therefore note the relevant  expressions here.  Since the geometric definition coincides with the matrix exponential for $\mathrm{SU}(1,1) $, the  results are essentially analytic continuations of the expressions  for $\mathrm{SU}(1,1) $ in the variable $\omega$. With the notation
\bee
\label{exponetialtilde}
\exp: \mathfrak{su}(1,1) \rightarrow \mathrm{SU}(1,1), \qquad   \widetilde{\exp}: \mathfrak{su}(1,1) \rightarrow \widetilde{\mathrm{SU}}(1,1), 
 \eee
for the exponential maps, we note  exponentials of typical timelike, spacelike and  lightlike elements in the Lie algebra. They give rise to, respectively, elliptic, hyperbolic and parabolic elements in the group, as exhibited  in \eqref{elliptic}, \eqref{parabolic} and \eqref{hyperbolic}. 
\begin{align}
\label{expcover}
\exp(-\alpha s^0) &= \begin{pmatrix} e^{i\frac{\alpha}{2}}  & 0  \\
0 & e^{-i\frac{\alpha}{2}} 
\end{pmatrix}, &  \widetilde{\exp}(-\alpha s^0) &= (\alpha,0),
\\
 \exp(-2\xi s^1 ) &=  \begin{pmatrix} \cosh \xi & \sinh\xi  \\
\sinh\xi &  \cosh \xi
\end{pmatrix},
  & 
 \widetilde{\exp}(-2\xi  s^1 )&= \left(0,\tanh \xi \right),
  \nonumber \\ 
\exp (2\zeta(s^0-s^1))&= \begin{pmatrix} 1-i\zeta & \zeta  \\
\zeta &  1+i\zeta
\end{pmatrix}, &
 \widetilde{\exp}(2\zeta(s^0-s^1)) &= \left(-2\tan^{-1}(\zeta) , \frac{\zeta }{1-i \zeta}\right).\nonumber 
\end{align}
Bearing in mind the comment made after \eqref{poincom}, we note in particular that  the adjoint action of  $(\alpha,0)$  on momentum space induces a mathematically positive rotation.

The unitary irreducible representations  (UIRs) of $\widetilde{\mathrm{SU}}(1,1)$ have been classified into a number of infinite families and are given in detail in  \cite{Sally} and also \cite{Grigore}. The particular class that is used to model anyonic particles  in \cite{JackiwNair,Plyushchay1,Plyushchay2}  is  called the discrete series.  We briefly review this here, using \cite{Sally} as a main reference  but adapting  notation used in \cite{Grigore}.  In particular, we label the UIRs in the discrete series  by $l \in \RR^+$ and a sign. 

The carrier spaces for the discrete series are given by the space of suitably  completed holomorphic ($+$) or anti-holomorphic ($-$) functions on the open unit disk. These spaces will be denoted $\mathcal{H}_{l\pm}$ or simply $\mathcal{H}$ when no confusion can arise. The completion is with respect to the appropriate Hilbert space inner product given below. For the family ($l+$),  the inner product between two holomorphic functions $f=\sum_{n=0}\alpha_nz^n$ and $g=\sum_{n=0}\beta_nz^n$  is 
\bee
\label{innerprodl}
\left( f,g \right)_l= \sum_{n=0}^{\infty}\frac{\Gamma(2l)\Gamma(n+1)}{\Gamma(2l+n)}\alpha_n\bar{\beta}_n, 
\eee
where $\Gamma$ is the Gamma function. The inner product above may be given an integral expression for the case $l>\frac{1}{2}$ as
\bee
\left( f,g \right)_{l\,>\frac{1}{2}}= \frac{2l-1}{\pi}\int_D (1-\abs{z}^2)^{2(l-1)}f(z)\bar{g(z)}\;\frac i 2 \dd z \wedge \dd \bar{z}.
\eee
For the representations labelled by $l-$ the inner products is as above except that the functions $f$ and $g$ are anti-holomorphic and therefore given as power series of the form $f=\sum_{n=0}\alpha_n\bar{z}^n$ and $g=\sum_{n=0}\beta_n\bar{z}^n$.

In order to express the action of $\widetilde{\mathrm{SU}}(1,1)$ on $\mathcal{H}$ for these representations we recall the projection map \eqref{coverproject} and \eqref{Momegamma}, and combine it 
with  a right action of $\mathrm{SU}(1,1)$ on the open unit disk:
\begin{equation}
z \cdot \bpm a & \bar{b} \\ b & \bar{a} \epm = \frac{az+b}{\bar{b}z+\bar{a}} = e^{i\omega} \frac{z+\gamma}{\bar{\gamma}z+1}.
\end{equation}

The action of $\widetilde{\mathrm{SU}}(1,1)$ on the carrier space $\mathcal{H}_{l+}$ is 
\bee
\label{PlusAnyonRep}
(D_{l+}(\omega, \gamma)f)(z)=e^{il\omega}(1-\abs{\gamma}^2)^{l}(1+\bar{\gamma}z)^{-2l}f(z\cdot \pi (\omega, \gamma)). 
\eee
On $\mathcal{H}_{l-}$  it is 
\bee
\label{MinusAnyonRep}
(D_{l-}(\omega, \gamma)f)(\bar{z})=e^{-il\omega}(1-\abs{\gamma}^2)^{l}(1+\gamma \bar{z})^{-2l}f(\overline{z\cdot \pi (\omega, \gamma)}).
\eee

The above action simplifies in the case $l \in \frac{1}{2}\NN$, where one recovers a genuine representation of the group  $\mathrm{SU}(1,1)$. In this case, one can use \eqref{abangles} to write 
\begin{equation}
(D_{l+}(\omega, \gamma)f)(z)=(\bar{a}+\bar{b}z)^{-2l}f(z\cdot \pi (\omega, \gamma)),
\end{equation}
 and 
 \begin{equation}
(D_{l-}(\omega, \gamma)f)(\bar{z})=(a+b\bar{z})^{-2l}f(\overline{z\cdot \pi (\omega, \gamma)}).
\end{equation}

A canonical choice of basis in $\mathcal{H}_{l +}$ is given by the orthonormal functions
\begin{equation}
e^+_{n,l}(z)=\left( \frac{\Gamma(2l+n)}{\Gamma(2l)\Gamma(n+1)} \right)^{\frac{1}{2}} z^n,\;\;n \in \NN,
\end{equation}
and for $\mathcal{H}_{l -}$ we have the basis
\begin{equation}
e^-_{n,l}(\bar{z})=\left( \frac{\Gamma(2l+n)}{\Gamma(2l)\Gamma(n+1)} \right)^{\frac{1}{2}} \bar{z}^n,\;\;n \in \NN.
\end{equation}
The state for $n=0$ in both $(l,+)$ and $(l,-)$ is the constant function. It plays an important role in constructing a covariant descriptions of  UIRs, so we introduce the notation
\bee
\label{groundstate}
\Ket {0}_l= e^+_{0,l}= e^-_{0,l}
\eee
for the map $z\mapsto 1$.

In these representations the infinitesimal generators of the Lie algebra, denoted $d_{l\pm}(s^{a})$, can be realised as  differential operators acting on the carrier space $\mathcal{H}$.  Starting with the usual definition
\begin{equation}
\label{dDrel}
d_{l\pm}(s^{a})= \left. \frac{\dd}{\dd \epsilon} \right \rvert_{\epsilon=0} D_{l\pm}(\widetilde{\exp}{\epsilon s^a}),
\end{equation}
we compute  for the positive discrete series 
\begin{align}
\label{InfAnyonRep+}
d_{l+}(s^{0})&=-il-iz\frac{\dd}{\dd z}, \nonumber \\
d_{l+}(s^{1})&= lz-\frac{1}{2}(1-z^2)\frac{\dd}{\dd z}, \nonumber  \\
d_{l+}(s^{2})&= ilz+\frac{i}{2}(1+z^2)\frac{\dd}{\dd z} ,
\end{align}
and for the negative series 
\begin{align}
\label{InfAnyonRep-}
d_{l-}(s^{0})&=il +i\bar{z}\frac{\dd}{\dd \bar{z}},  \nonumber \\
d_{l-}(s^{1})&=l\bar{z}-\frac{1}{2}(1-\bar{z}^2)\frac{\dd}{\dd\bar{z}},  \nonumber\\
d_{l-}(s^{2})&=-il\bar{z}-\frac{i}{2}(1+\bar{z}^2)\frac{\dd}{\dd\bar{z}}.
\end{align}
For later use, we note that the linear combination  $s^{a}p_a$ for an arbitrary vector $\bp$ acts according to 
\begin{equation}
\label{dl+}
d_{l+}(s^{a})p_a=-il p_0 +i(p_2-ip_1)lz  +i\left(\frac{1}{2}(p_2+ip_1)- p_0z+\frac{1}{2}(p_2-ip_1)z^2\right)\frac{\dd}{\dd z},
\end{equation}
and
\begin{equation}
\label{dl-}
d_{l-}(s^{a})p_a=il p_0 -i(p_2+ip_1)l\bar{z} -i\left(\frac{1}{2}(p_2-ip_1)- p_0\bar{z}+\frac{1}{2}(p_2+ip_1)\bar{z}^2\right)\frac{\dd}{\dd\bar{z}}.
\end{equation}

The vector fields appearing in the action of the $s^{a}p_a$ have a natural geometrical interpretation which is  familiar in the context of the mini-twistor  correspondence between points in Euclidean 3-space and the set of all the lines through that point.  We explain briefly how this point of view fits into our Lorentzian setting. 

We can parametrise the set of all timelike lines in 2+1 dimensional Minkowski space in terms of a timelike vector $\bq$,  normalised so that $\bq^2=1$ and giving the direction of the line, and a vector $\bk$ which lies on the line and which can  be chosen to satisfy $\bk\cdot \bq=0$ without loss of generality. Geometrically, $\bq$ lies on the  two-sheeted  hyperboloid, and $\bk$ lies in the tangent space at $\bq$.

 In terms of an orthonormal basis $\be^{0},\be^{1},\be^{2}$ of Minkowski space,  and the complex linear combination $\be=\be^{1}+i\be^{2}$,  we can parametrise $\bq$ in terms of a complex variable $z$ in the unit disk via
\[
\bq=\frac{1}{1-|z|^2}\left( (1+|z|^2)\be^{0} +z\bar{\be} + \bar{z} \be\right),
\]
and the tangent vector  $\bk$ in terms of a complex number $w$ via
\[
\bk = w \frac{\partial  \bq}{\partial  z} + \bar{w} \frac{\partial  \bq}{\partial  \bar z} = \frac{1}{(1-|z|^2)^2} \left( 2(w \bar z + \bar w z) \be^{0} + (\bar w + w \bar z^2)\be +  (w + \bar w z^2)\bar \be\right).
\]
Then one checks that the point $\bp=p_0\be^{0} + p_2\be^{1} +p_1\be^{2}$ lies on the line through $\bk$ and in the direction $\bq$ if and only if
\bee
w= \frac 12 (p_2+ip_1) - p_0z + \frac 12 (p_2-ip_1) z^2.
\eee
In other words,  the derivatives appearing in \eqref{dl+} and \eqref{dl-} are precisely the holomorphic and anti-holomorphic part of the tangent vector which characterises a line containing $\bp$ and in the direction $\bq$ determined by $z$.

\subsection{Massive $\tilde{P}_3$ representations}
\label{P3 Reps}
There are various ways of getting from the UIRs of the Poincar\'e group to the covariant wave equations of relativistic physics.
In  \cite{SemiDual,SchroersWilhelm}, a procedure was developed which is also effective when the Poincar\'e group is deformed to the quantum double of the Lorentz group or one of its covers. We briefly review the method here in a convenient form for extension to the anyonic case. However, relative to \cite{SchroersWilhelm} we change the sign convention for mass and spin to agree with the one used in \cite{BaisScatt}.

Using the isomorphism $\tilde{P}_3 \simeq \mathrm{SU}(1,1) \ltimes \mathfrak{su}(1,1)^*$ in \eqref{P3Iso},  UIRs of $\tilde{P}_3$ may be classified by the adjoint orbits of $\mathrm{SU(1,1)}$ in momentum space $\mathfrak{su}(1,1)$  together with an UIR of associated stabiliser groups.  For massive particles, the former encodes the mass, and the latter its spin. In the case of a particle of mass $m\neq 0$, we denote  the adjoint orbit  by $ O^T_m$. It elements are obtained by boosting the representative momentum $-\lambda ms^0$ (which  generates a  mathematically positive rotation)  to obtain the typical element 
\bee
\label{signconventions}
p=   -v\lambda ms^0v^{-1}=-\lambda \bp\cdot\bs.
\eee
The coordinate vector $\bp$ satisfies $\bp^2 =m^2$ and  the sign constraint $mp_0>0$ so that $\bp$  lies in the  $p_0>0$ half space for positive $m$,  and in the  $p_0<0$ half space for negative  $m$. Thus we have two equivalent characterisations of the adjoint orbit corresponding to massive particles:
\bee
O^T_m = \Set{ \Ad_v(-\lambda ms^0) | v\in \mathrm{SU}(1,1)}  = \Set{- \lambda p_as^a \in \mathfrak{su}(1,1) | \bp^2=m^2, mp_0>0 }.
\eee
 Geometrically, this is  the  upper  ($m>0$) or lower ($m<0$) sheet of the two-sheeted hyperboloid. Other types of adjoint orbits include  the trivial orbit $\{ 0\}$ describing the vacuum, the   forward or backward lightcone describing massless particles and  single-sheeted hyperboloids  describing  tachyons, but  we will not consider these here.

The associated stabiliser group $N^T$ of $O^T_m$ is given by
\bee
N^T=\Set{ v \in \mathrm{SU}(1,1)  | \Ad_v(s^0)= s^0)} = \Set{\exp (-\phi s^0) | \phi \in [0,4\pi) } \nonumber \\
\simeq U(1).
\eee
The UIRs of $U(1)$ are  one-dimensional and labelled by $s\in \frac{1}{2}\mathbb{Z}$  in our conventions.

The carrier space of the UIRs   of $\tilde{P}_3$ for massive particles with spin $s$  can be given in  two equivalent ways. Either one considers functions on $\mathrm{SU}(1,1)$ which satisfy an equivariance condition or sections of associated vector bundles over the homogeneous space $\mathrm{SU}(1,1)/N^T \simeq O^T_m$. We focus on the former method here but, refer the reader to \cite{BaisScatt} for a discussion of their equivalence in the context of 3d gravity and to \cite{Barut} for a general reference. Adopting the conventions of \cite{BaisScatt}, we   define the carrier space as
\begin{align}
\label{P3 Carrier Sp}
V_{ms}= \left \lbrace \psi \colon \mathrm{SU}(1,1) \rightarrow \CC  \vert \psi \left(ve^{-\alpha s^0}\right)=e^{-is\alpha}\psi(v) \; \forall \,\alpha \in  [0,4\pi), v \in  \mathrm{SU(1,1)}, \right. \nonumber  \\
\left. \int_{\mathrm{SU}(1,1)/N^T}\abs{\psi}^2 d  \nu < \infty  \right\rbrace,
\end{align}
where $d \nu$ is the invariant measure on the coset $\mathrm{SU}(1,1)/N^T$.
The action of an element $(g,a)\in \tilde{P}_3$  on $\psi \in V_{ms}$ is 
\begin{equation}
\label{pi ms rep}
(\pie_{ms}(g,a)\psi )(v)=\exp\left(i\langle a,\Ad_{g^{-1}v}(-\lambda ms^0)\rangle \right)\psi(g^{-1}v),
\end{equation}
where, in accordance with \eqref{P3Iso}, $a$ is interpreted as an element of  $\mathfrak{su}(1,1)^*$, and   $\langle \cdot , \cdot\rangle$  is the   pairing between elements of $\mathfrak{su}(1,1)^*$ and $\mathfrak{su}(1,1)$ introduced and discussed in Sect.~\ref{poincareconventions}. We have attached the superscript `eq' to distinguish this  equivariant formulation from the later covariant version. 

\subsection{Covariant field representations}
In field theory we do not usually work with the space of equivariant functions as just described. Instead we use covariant fields 
\begin{equation}
\phi \colon \RR^{2,1} \rightarrow V,
\end{equation}
where $V$ is a carrier space for a (usually finite dimensional) representation of the Lorentz group. In general such fields do not form irreducible representations of the Poincar\'e group and, as a result, additional constraints need to be imposed to achieve this. For fields defined on momentum space these constraints are  algebraic,  but after Fourier transform they  yield the familiar wave equations for a field of  definite spin.

Following  \cite{SchroersWilhelm} for the method, but changing the sign convention to agree with \cite{BaisScatt},  we construct a covariant field
\begin{align}
\tilde{\phi} \colon O^T_m \rightarrow \CC^{2\abs{s}+1} 
\end{align} 
from a given $\psi \in V_{ms}$ via 
\begin{equation}
\tilde{\phi}(p)=\psi(v) \rho ^{\abs{s}}(v)\Ket {\abs{s},-s}, \qquad \text{with} \quad 
p=-\lambda m vs^0v^{-1}.
\end{equation}
Here $\{ \Ket{\abs{s},k}: k = -|s|,-|s|+1, \ldots |s|-1, |s|\}$  is a  basis of the (non-unitary) $(2|s|+1)$-dimensional  representation of $\mathfrak{su}(1,1)$, satisfying 
\bee
\rho^{|s|}(s^0)  \Ket{\abs{s},k}=ik  \Ket{\abs{s},k},
\eee
see \cite{SchroersWilhelm} for details.  

To check that $\tilde{\phi}(p)$ is well-defined, one  needs to show that 
\begin{equation}
\psi(v) \rho ^{\abs{s}}(v)\Ket {\abs{s},-s}=\psi(vn) \rho ^{\abs{s}}(vn)\Ket {\abs{s},-s}\;\forall n \in N^T,
\end{equation}
but this is true because the phase picked up by $\psi$ under the action of an element $n$ of the stabiliser is precisely cancelled by the action of $\rho ^{\abs{s}}(n)$ on the state $\Ket {\abs{s},-s}$. This construction works for massive particles since  $\rho ^{\abs{s}}(s^0)$ has imaginary eigenvalues. However,  this is not the case for the momentum representatives on massless and tachyonic orbits and hence the above procedure is limited to  particles with timelike  momentum. 

Adapting the results in \cite{SchroersWilhelm} to our conventions, the covariant field $\tilde{\phi}$ necessarily satisfies the condition
\begin{equation}
\label{SpinConstraint}
\left(i\rho^{\abs{s}}(s^a)p_a-ms\right)\tilde{\phi}(p)=0,
\end{equation}
which we call the spin constraint.  In order to carry out the envisaged Fourier transform, we would like to  extend  $\tilde\phi$   to a function on all of  the linear momentum space $\mathfrak{su}(1,1)$. However, we then need to impose the  mass constraint $\bp^2=m^2$  and the sign constraint  $mp_0 >0$
to ensure that  $\tilde \phi$ has support on the orbit $ O^T_m$. The sign constraint makes sense  when the mass constraint is enforced since $m\neq 0$ by assumption and  therefore $(p_0)^2 >0$.  To  restrict the support of $\tilde \phi$  to the `forward' mass shell when $m>0$ and to the `backward' mass shell when $m<0$, we use the Heaviside function $\Theta$ and 
 define the carrier space
\begin{equation}
\label{covcarrier}
W_{ms}=\Set{ \tilde{\phi} \colon \mathfrak{su}(1,1) \rightarrow \CC^{2\abs{s}+1} | (i\rho^{\abs{s}}(s^a)p_a-ms)\tilde{\phi}(p)=0, (\bp^2-m^2)\tilde{\phi}(p)=0, \Theta( -mp_0)\tilde \phi=0}.
\end{equation}
In the corresponding definition in  \cite{SchroersWilhelm}, the sign constraint was not  included, but for us this  inclusion is convenient because we directly obtain a UIR of  $\tilde{P}_3$ without adding further conditions. 
The action  of  $(g,a)\in \tilde{P}_3$  on this space is  given by 
\begin{equation}
(\pic(g,a) \tilde{\phi})(p)=\exp(i\langle a,\Ad_{g^{-1}}p\rangle)\rho ^{\abs{s}}(g)\tilde{\phi}(\Ad_{g^{-1}}p),
\end{equation}
which we call the covariant formulation.

In \cite{SchroersWilhelm} it is also shown that the above covariant fields produce UIRs of $\tilde{P}_3$ for the familiar cases of spin $s=0,\frac{1}{2}, 1$ and that the mass constraint for spin zero and the spin constraints for $s=\frac{1}{2}$ and $s=1$ produce the momentum space versions of the Klein-Gordon equation, Dirac equation  and of  field equations which square to the  Proca equation\footnote{The spin 1 equation was simply called Proca equation  in \cite{SchroersWilhelm} but  it is  more precisely  a first order equation  which implies the Proca equation.  Its relation to self-dual massive field theory is discussed in \cite{Gitman}}.

We refer the reader to \cite{SchroersWilhelm} for details of the Fourier transform of the spin constraints to relativistic field equations in spacetime. We now turn to the anyonic case, where we will discuss both the covariant formulation of the UIRs and the Fourier transform.

\section{Anyonic wave equations}
Anyons are  quantum particles with fractional spin which occur in systems confined to two spatial dimensions. In the relativistic case, the  theoretical possibility of anyonic particles is a consequence of the infinite connectedness of the Lorentz group $\Lor$.

To describe relativistic anyons we need to consider the representation theory of the universal cover of the Poincar\'{e} group $P^\infty_3$. The UIRs are classified in \cite{Grigore} using the method of induced representations. The action  of $ (\omega, \gamma)\in \widetilde{\mathrm{SU}}(1,1)$ on momentum space is 
 the adjoint action $\Ad_{(\omega, \gamma)}$.
The stabiliser  group for a massive particle, with standard momentum $-\lambda ms^0$, is therefore
\bee
\tilde{N}^T = \Set{(\omega, \gamma) \in \widetilde{\mathrm{SU}}(1,1) | \Ad_{(\omega, \gamma)}s^0=s^0 } = \left\lbrace (\omega, 0) \in \widetilde{\mathrm{SU}}(1,1) \right\rbrace  \simeq \RR.
\eee
The one-dimensional UIRs of the stabiliser are labelled by $s \in \RR$ which  represents the spin of the massive particle.  With our results \eqref{expcover} for the exponential map into $\widetilde{\mathrm{SU}}(1,1)$ and using  \eqref{Momegamma}, we note that 
\bee
(\omega,\gamma)\widetilde{\exp}(-\alpha s^0)= (\omega+\alpha, \gamma).
\eee
Thus, we have  the following equivariant description of the carrier space:
\begin{align}
\label{AnyonicEquivariance}
V_{ms}^A= \left\lbrace \psi \colon \widetilde{\mathrm{SU}}(1,1) \rightarrow \CC \vert \psi (\omega+\alpha, \gamma)=e^{-is\alpha} \psi(\omega,\gamma)\; \forall \alpha \in  \RR, \forall (\omega,\gamma) \in \widetilde{\mathrm{SU}}(1,1),
\right.  \nonumber  \\
\left. \int_{\widetilde{\mathrm{SU}}(1,1)/\tilde{N}^T}\abs{\psi}^2 d\nu < \infty  \right\rbrace.
\end{align}
 The action of $((\omega,\gamma), a)\in P
 ^\infty_3$ on the space $V_{ms}^A$ is 
\begin{equation}
\label{AnyonicUIR}
\left(\pie_{ms}((\omega, \gamma),a)\psi \right)(v)=\exp\left(i\langle a,\Ad_{((\omega, \gamma)^{-1}v)}(-\lambda ms^0)\rangle \right)\psi \left((\omega, \gamma)^{-1}v\right).
\end{equation}
We now follow the procedure of the previous section to construct  anyonic covariant fields.

\begin{definition}[Anyonic Covariant Field] The anyonic covariant field $\tilde{\phi}_{\pm}$ associated to an equivariant field $\psi\in V^A_{ms}$ is the   map
\begin{align}
\tilde{\phi}_{\pm} \colon O^T_m \rightarrow \mathcal{H}_{l \pm},
\end{align} 
where $\mathcal{H}_{l \pm}$ is the carrier space for the discrete series representations of $\widetilde{\mathrm{SU}}(1,1)$ given in \eqref{SL2 Coverings},
defined via
\begin{equation}
\label{AnyonicCovField}
\tilde{\phi}_{\pm}(p)=\psi(\omega, \gamma) D_{l\pm}(\omega, \gamma)\Ket {0}_l.
\end{equation}
Here  $(\omega,\gamma)\in\widetilde{\mathrm{SU}}(1,1)$ is chosen so that  $p=-\lambda m \,\Ad_{(\omega, \gamma)}(s^0)$ and 
$s=l$ for $\tilde{\phi}_+$ and $s=-l$ for $\tilde{\phi}_-$. 
As before, $D_{l\pm}$ are the discrete series representations of $\widetilde{\mathrm{SU}}(1,1)$ and $\Ket {0}_l$ is defined in \eqref{groundstate}.
\end{definition}
Thus we should use $\tilde{\phi}_+$ to describe positive spin particles and $\tilde{\phi}_-$ for negative spin.
\begin{lemma} The anyonic covariant fields \eqref{AnyonicCovField} are well defined.
  \end{lemma}

\begin{proof} 
One needs to check that this definition is independent of the choice of $(\omega,\gamma)$, i.e., we require
\begin{equation}
\psi(\omega,\gamma) D_{l\pm}(\omega,\gamma)\Ket {0}_l=\psi((\omega,\gamma)(\alpha,0)) D_{l\pm}\left((\omega,\gamma)(\alpha,0)\right)\Ket {0}_l \qquad \text{for all} \quad \alpha \in \RR.
\end{equation}
Expanding the right hand side and using the equivariance of $\psi$, and the action of the stabiliser subgroup elements in the  representations $D_{l\pm}$  on the vacuum  state one obtains
\begin{align}
\label{invariance}
\psi((\omega,\gamma) (\alpha,0)) D_{l\pm}\left((\omega,\gamma)(\alpha,0)\right)\Ket {0}_l &=\psi((\omega+\alpha,\gamma)  D_{l\pm}\left((\omega,\gamma)(\alpha,0)\right)\Ket {0}_l \nonumber \\
&=e^{-is\alpha} \psi(\omega,\gamma)D_{l+}(\omega,\gamma)D_{l+}(\alpha, 0))\Ket {0}_l\nonumber \\
&=e^{-is\alpha} \psi(\omega,\gamma)D_{l+}(\omega,\gamma)e^{il\alpha}\Ket {0}_l \nonumber \\
&=e^{i(l-s)\alpha}\psi(\omega,\gamma)D_{l+}(\omega,\gamma)\Ket {0}_l .
\end{align}
Hence we obtain invariance if $l-s=0$, as claimed. An analogous argument applied to $\tilde{\phi}_-$  shows that $l+s=0$ is required in that case.
\end{proof}

The anyonic covariant field carries a unitary representation of the Poincar\'e group, which we again denote $\pic$:
\begin{equation}
\label{AnyonicCovAction}
(\pic_{ms}((\omega, \gamma), a) \tilde{\phi}_{\pm})(p)=\exp(i\langle a,\Ad_{(\omega, \gamma))^{-1}}p\rangle)D_{l\pm}((\omega, \gamma))\tilde{\phi}_{\pm}(\Ad_{(\omega, \gamma)^{-1}}p).
\end{equation}
Without further condition, this representation is not irreducible. To achieve irreducibility, we need 
 anyonic versions of  the  spin, mass and sign constraints. Our version of the spin constraint is  the equation considered by Plyushchay in \cite{Plyushchay1,Plyushchay2}. As pointed out in those papers, it is also essentially   a dimensional reduction of the Majorana equation \cite{Majorana}.

\begin{lemma}[Anyonic Spin Constraint]
The anyonic covariant fields $\tilde \phi_{\pm}$ satisfy the following constraints
\bee
\label{AnyonSpinCon}
(d_{l\pm}(p)-i\lambda ms) \tilde{\phi}_{\pm}(p)=0,
\eee
where $d_{l\pm}$ is the Lie algebra representation associated to the discrete series representation $D_{l\pm}$ via \eqref{dDrel}, with  $s=l$  for the positive series  and $s=-l$ for the negative series. 
\end{lemma}

\begin{proof} With  
$p=-\lambda p_as^a=-\lambda m\Ad_{(\omega, \gamma)}s^0 $, we compute for  $\tilde{\phi}_+$
\begin{align}
d_{l+}(p) \tilde{\phi}_+(p)&=-\lambda md_{l+}(\Ad_{(\omega, \gamma)}s^0)\tilde{\phi}_+(p) \nonumber \\
&=-\lambda  m\psi(\omega, \gamma)D_{l+}(\omega, \gamma)d_{l+}(s^0)D_{l+}((\omega, \gamma)^{-1})D_{l+}(\omega, \gamma)\Ket {0}_l \nonumber \\
&= -\lambda m\psi(\omega, \gamma)D_{l+}(\omega, \gamma)d_{l+}(s^0)\Ket {0}_l \nonumber \\
&=i\lambda ml  \psi(\omega, \gamma)D_{l+}(\omega, \gamma)\Ket {0}_l \nonumber \\
&= i  \lambda ms\tilde{\phi}_+(p).
\end{align}
An analogous computation for $\tilde{\phi} _-$, noting that $s=-l$, completes the proof.
\end{proof}

Note that, in components and with the sign convention \eqref{signconventions}, the spin constraint  takes the form
\bee
\label{convenspin}
(i d_{l\pm}(s^a)p_a -  ms) \tilde{\phi}_{\pm}(p)=0,
\eee
which has the same form as the finite-component version \eqref{SpinConstraint}.

In order to  construct UIRs of $P^\infty_3$ in terms of covariant fields  we define
 the carrier space which generalises  \eqref{covcarrier} to infinite-component fields:
\begin{equation}
\label{WAms}
W_{ms}^A=\Set{ \tilde{\phi}_{\pm} \colon \mathfrak{su}(1,1) \rightarrow \mathcal{H}_{l\pm} |  \left( d_{l\pm}(p) - i\lambda ms\right)\tilde{\phi}_{\pm}(p)=0, (\bp^2-m^2)\tilde{\phi}_{\pm}(p)=0 },
\end{equation}
where we choose the upper sign if $s>0$ and the lower sign if $s<0$.  

Note that, unlike in the finite-dimensional case,  we do not also need to impose  the  constraint that $p^0$ and $m$ have the same sign. This follows from our conventions \eqref{signconventions} and  from the fact that $id_{l+}(s^0)$ has only positive eigenvalues and  $i d_{l-}(s^0)$ only negative eigenvalues. This property was one of the motivations for Majorana to construct his infinite-component fields in \cite{Majorana}. 

The  space $W_{ms}^A$ is a Hilbert space with the inner product
\bee
(\tilde{\phi}_1, \tilde{\phi}_2) = \int_{O^T_{m}}(\tilde \phi_1(p),\tilde\phi_2(p))_l\, d \nu,
\eee
where $d\nu$ is the invariant measure on the hyperboloid $O^T_m$, and $(\cdot,\cdot)_l$ is the inner product  \eqref{innerprodl}  on $\mathcal{H}_{l+}$, with an analogous expression for fields taking values in  $\mathcal{H}_{l-}$.

We now show that the representations $V^{A}_{ms}$ and $W^{A}_{ms}$ are isomorphic.  This implies that the covariant fields subject to the mass and  spin constraints form UIRs  of the universal cover  $P^\infty_3$ of the Poincar\'e group. 
  
\begin{theorem}[Irreducibility of the carrier space $W^{A}_{ms}$] 
\label{intertwining} The  covariant representation $\pic_{ms}$ of  $P^\infty_3$  on $W^{A}_{ms}$  defined in \eqref{AnyonicCovAction} is unitarily equivalent to the  equivariant  representation $\pie_{ms}$ on $V^{A}_{ms}$ defined in \eqref{AnyonicUIR}.  In particular, it is therefore irreducible.
\end{theorem}

\begin{proof}
We claim that the following maps  are intertwiners:
\begin{equation}
L_{\pm} \colon V^{A}_{ms} \rightarrow W^{A}_{ms},\qquad (L_{\pm}(\psi))(p)= \psi(\omega, \gamma) D_{l\pm}(\omega, \gamma)\Ket {0}_l,
\end{equation}
where, as before, $(\omega,\gamma)\in\widetilde{\mathrm{SU}}(1,1)$ is chosen so that  $p=-\lambda m \Ad_{(\omega, \gamma)}(s^0)$ and $s=l$ for the positive series and $s=-l$ for the negative series. We have already shown that
$\psi(\omega, \gamma) D_{l\pm}(\omega, \gamma)\Ket {0}_l$ only depends on $p$, satisfies the spin constraint and has support entirely on the orbit $O^T_m$, so that the mass  constraint is  also satisfied. Thus the maps $L_\pm$ are well-defined.

The maps $L_{\pm}$ are injective because of the unitarity of $D_{l\pm}$.  To show that they are  surjective,  we pick $\tilde \phi \in W^{A}_{ms}$ and construct a preimage.  Focusing on $D_{l+}$, and noting that 
\begin{equation}
d_{l+}(p)=-\lambda m D_{l+}(\omega, \gamma)d_{l+}(s^0)D_{l+}((\omega, \gamma)^{-1}),
\end{equation}
 the spin constraint  \eqref{AnyonSpinCon} is equivalent to
 \bee
  d_{l+}(s^0)D_{l+}((\omega, \gamma)^{-1})\tilde \phi (p) = -is D_{l+}((\omega, \gamma)^{-1})\tilde \phi (p).
 \eee
Recalling that $l=s$, comparing with \eqref{InfAnyonRep+}, and recalling that  $\ket{0}_l$ is, up to a factor, the unique solution of 
$   d_{l+}(s^0)f = -il  f$, we deduce the proportionality  
\bee
D_{l+}((\omega, \gamma)^{-1})\tilde \phi (p) = \psi (\omega,\gamma) \ket{0}_l,
\eee
where the proportionality  factor $\psi$ may depend on $(\omega,\gamma)$. Moreover, it must have the property
\bee
\psi (\omega+\alpha,\gamma) = e^{-i\alpha s} \psi(\omega,\gamma), 
\eee
to ensure independence of  the choice of $(\omega,\gamma)$  for given $p$, since
\begin{align}
\psi (\omega+\alpha,\gamma) \ket{0}_l &=  D_{l+}((\omega+\alpha, \gamma)^{-1})\tilde \phi (p) \nonumber  \\
& = D_{l+}(-\alpha,0)D_{l+}((\omega, \gamma)^{-1})\tilde \phi (p) \nonumber \\
&= D_{l+}(-\alpha,0)\psi (\omega,\gamma) \ket{0}_l \nonumber \\
&=e^{-i\alpha s} \psi (\omega,\gamma) \ket{0}_l.
\end{align}
Thus $\psi \in V^A_{ms}$ and $L_+(\psi) =\tilde \phi$, as required to show that $L_+$ is surjective.
An analogous argument for $L_{-}$ shows that both  $L_{\pm}$  are bijections. 

The  intertwining property is equivalent to  the commutativity of the diagram
\bee
\begin{tikzcd}[row sep=huge, column sep = large]
V^A_{ms} \arrow{r}{\quad \pie_{ms}((\omega, \gamma),a) \quad} \arrow[swap]{d}{L_{\pm}} & V^A_{ms} \arrow{d}{L_{\pm}} \\
W^A_{ms} \arrow{r}{\quad \pic_{ms} ((\omega, \gamma),a)\quad } & W^A_{ms}
\end{tikzcd},
\eee
for $((\omega, \gamma),a)\in P^\infty_3$. This is a straightforward calculation based upon the maps $L_{\pm}$, and the actions given in \eqref{AnyonicUIR} and \eqref{AnyonicCovAction}.

The unitarity of $L_{\pm}$ follows from the unitarity of $D_{l\pm}$, since 
\begin{align}
(L_\pm\psi_1, L_\pm \psi_2) &= \int_{O^T_{m}}(\psi_1(\omega,\gamma) D_{l\pm} (\omega,\gamma)\ket{0}_l, \psi_2(\omega,\gamma) D_{l\pm} (\omega,\gamma)\ket{0}_l)_l
\, d \nu \nonumber \\
& = \int_{O^T_{m}}  \psi_1(\omega,\gamma)  \bar{\psi}_2(\omega,\gamma)\, d \nu.
\end{align}
\end{proof}

Finally, we Fourier transform the spin constraint in the form \eqref{convenspin}  to obtain the  anyonic wave equation promised in the title of this section. Since the field $\phi$ lives on the Lie algebra $\mathfrak{su}(1,1)$, its Fourier transform should live on $\mathfrak{su}(1,1)^*$, i.e., the Fourier transform is a map
\bee
\label{flatfourier}
L^2(\mathfrak{su}(1,1), \mathcal{H}_{l\pm})\rightarrow L^2(\mathfrak{su}(1,1)^*, \mathcal{H}_{l\pm}). 
\eee
Using the terminology introduced  after  equation \eqref{p3algebra}, we  expand $x\in \mathfrak{su}(1,1)^*$ and $p\in \mathfrak{su}(1,1)$  as 
\bee
x=x^aP_a, \quad p= -\lambda p_as^a. 
\eee
Then, with the pairing  given  in \eqref{pairing},  
we define  the Fourier transform of $\tilde \phi \in W_{ms}^A$ by $\phi$ 
as 
\bee
\phi_{\pm}(x)=\int_{\mathfrak{su(1,1)}}
e^{i\langle x,p\rangle}\tilde \phi_{\pm}(p)\; d^3 \bp= \int_{\RR^3}
e^{i\bx\cdot\bp}\tilde \phi_{\pm}(p)\; d^3 \bp.
\eee

The field  $\phi_\pm$ satisfies the Klein-Gordon equation
\bee
\label{KG}
(\partial_0^2 -\partial_1^2-\partial_2^2+m^2)\phi_{\pm} =0,
\eee
 by virtue of the mass constraint. The spin constraint  implies the  following first order equation: 
\bee
\label{anywave1}
\left(d_{l\pm}(s^a)\partial_a - ms \right)\phi_{\pm}=0,
\eee
where we wrote $\partial_a=\partial/\partial x^a$. 
Using the explicit forms \eqref{dl+}  and \eqref{dl-} of $d_{l\pm}$ with $l=s$ for the positive series  and $l=-s$ for the negative series, 
and with the abbreviations 
\bee
\partial=\frac 12\left(\partial_2 -i \partial_1\right), \quad 
\bar{\partial}=\frac 12\left(\partial_2 +i \partial_1\right),
\eee
the  equation \eqref{anywave1} can also  be written as  
\bee
\label{anywave2+}
(-is\partial_0 +2isz\partial  +i\left(\bar{\partial}- z\partial_0+ z^2 \partial \right)\frac{\dd}{\dd z} -ms)\phi_+=0,
\eee
and 
\bee
\label{anywave2-}
(-is\partial_0 +2is\bar{z}\bar{\partial}  -i\left(\partial- \bar{z}\partial_0+ \bar{z}^2 \bar{\partial }\right)\frac{\dd}{\dd\bar z} -ms)\phi_-=0.
\eee

The anyonic relativistic wave equation we have constructed  for arbitrary spin $s \in \RR$ can be viewed in two ways: either as a partial differential equation in Minkowski space for a field taking values in any infinite-dimensional Hilbert space, as suggested by the formulation \eqref{anywave1}, or as a partial differential equation for a field on the product of Minkowski space and the hyperbolic disk, as emphasised in the formulation \eqref{anywave2+} and \eqref{anywave2-}.

It is worth stressing that, for  $s \in \frac{1}{2} \NN$, our anyonic equation \eqref{anywave1}  does not reduce to the equation  \eqref{SpinConstraint} for an equivariant field with finitely many components. These equations are not equivalent as they are characterised by different irreducible representations.

\section{Gravitising anyons}
\label{anywave}
\subsection{The  Lorentz double  and its representations}
We now extend and  apply our  method for deriving wave equations from   Lorentz covariant UIRs of the Poincar\'e group to a deformation of the Poincar\'e symmetry to  the quantum double of the  universal cover of the  Lorentz group, or Lorentz double for short. 
As reviewed in our Introduction, this is motivated by results from the study of 3d gravity and a general interest in understanding possible quantum deformations of standard wave equations. Referring to  \cite{Lessons3DGrav} and \cite{NonCommutSch} for reviews, 
we sum up evidence for the emergence of quantum doubles in the quantisation of 3d gravity.

\noindent {\em Deformation of Poincar\'e symmetry:}  As explained in \cite{EuclidSchroers}  and \cite{BaisScatt} for the, respectively, Euclidean and Lorentzian case, the quantum double of the rotation and Lorentz group is a deformation of the group algebra of, respectively, the Euclidean and Poincar\'e group. 

\noindent {\em Gravitational scattering:}   The $R$-matrix of the Lorentz double can be used to derive a universal scattering cross section for massive particles with spin by treating   gravitational  scattering  in 2+1 dimensions as a non-abelian Aharonov-Bohm scattering process \cite{BaisScatt}.  This universal scattering cross section agrees with previously computed  special cases, like the quantum scattering of a  light spin 1/2 particle on the conical spacetime generated by a heavy massive particle, in suitable limits  - see \cite{tHooft,DeserJackiw}.

\noindent {\em Combinatorial quantisation:} The quantum double of the Lorentz group arises naturally in the combinatorial quantisation of the Chern-Simons formulation of 3d gravity with vanishing cosmological constant. The classical limit of the quantum $R$-matrix is a classical $r$-matrix which is compatible with the  non-degenerate bilinear symmetric and invariant pairing used in the Chern-Simons action \cite{EuclidSchroers,BaisScatt}, and the Hilbert space of the quantised theory can be constructed from unitary representations of the Lorentz double \cite{MeusburgerSchroers1,CMBSCombQuant}.

\noindent {\em Independent derivations:} Quantum  doubles also emerges in approaches to 3d quantum gravity which do not rely on  the combinatorial quantisation programme.  In \cite{LQGQDouble} the quantum double  is shown to play the role of quantum symmetry in 3d loop quantum gravity. In \cite{FL} it appears in  a path integral approach to 3d quantum gravity.  

In analogy with our treatment of the Poincar\'e group in 2+1 dimensions, we consider the double cover $\mathrm{SU}(1,1)$ and the universal cover $\widetilde{\mathrm{SU}}(1,1)$ of the identity component of the Lorentz group. Our goal is to obtain a deformation of the wave equation by covariantising and then Fourier transforming, in a suitable sense, the UIRs of  the quantum double of $\widetilde{\mathrm{SU}}(1,1)$. This extends the results obtained in \cite{SchroersWilhelm} for the double cover $\mathrm{SU}(1,1)$. As we shall see, the universal cover is technically more involved but also conceptually more interesting.

The quantum double of a Lie group can be defined in several ways. We follow \cite{QuantDoubBaisMull,QDLocComp} with the conventions used in \cite{SchroersWilhelm,SemiDual}. In this approach we view the quantum double $\mathcal{D}(G)$ as the  Hopf algebra which,  as a vector space,  is the space of continuous complex valued functions $C(G\times G)$.  In order to exhibit the full Hopf algebra structure we need to adjoin singular $\delta$-distributions. 

The Hopf algebra structure for $\mathcal{D}(G)$ with product $\bullet$, co-product $\Delta$, unit $1$, co-unit $\epsilon$, antipode $S$, $*$-structure and ribbon element $c$ is then as follows:
\begin{align}
&(F_1 \bullet F_2)(g,u) =\int_{G}F_1(v,vuv^{-1})F_2(v^{-1}g,u)d v, \nonumber \\
&1(g,u)= \delta_e(g), \nonumber \\
&(\Delta F)(g_1,u_1;g_2,u_2)=F(g_1,u_1u_2)\delta_{g_1}(g_2), \nonumber \\
&\epsilon(F)=\int_G F(g,e) d g, \nonumber \\
&(SF)(g,u)=F(g^{-1}, g^{-1}u^{-1}g), \nonumber \\
&F^*(g,u)=\overline{F(g^{-1}, g^{-1}ug)}. \nonumber \\
&c(g,u)=\delta_g(u),
\end{align}
where we write $dg$ for the left Haar measure on the group.

The representation theory of the double is given in \cite{QDLocComp}. In the case of $\mathcal{D}(\mathrm{\mathrm{SU(1,1)})}$,  the UIRs are classified by  conjugacy classes in $\mathrm{SU(1,1)}$ and  UIRs of associated stabiliser subgroups. This should be viewed as a deformation of the discussion of $\tilde{P}_3$, where  we had adjoint orbits in the linear momentum space. In the gravitational case, the group itself is interpreted as momentum space and orbits are conjugacy classes. Here we encounter the idea of curved momentum space discussed in the outline. 

As we are interested in the case of massive particles we will only give the analogue of the massive representations of $\tilde{P}_3$, and refer to   \cite{QDLocComp} and  \cite{SchroersWilhelm} for the complete list.  By definition, elements of $\mathrm{SU(1,1)}$  have eigenvalues which multiply to $1$. In the case where these eigenvalues are complex conjugates, one has two disjoint families of elliptic conjugacy classes labelled by an angle $\theta$:
\begin{equation}
E(\theta)=\Set{v\exp \left(-\theta s^0\right)v^{-1}|v\in \mathrm{\mathrm{SU(1,1)}},\;\theta \in (0,2\pi)\cup (2\pi,4\pi)}.
\end{equation}
The stabiliser subgroup of the representative element $\exp (-\theta s^0)$ in $E(\theta)$ is
\begin{equation}
N^T=\Set{\exp (-\alpha s^0)|\alpha \in [0,4\pi)}\simeq U(1).
\end{equation}

In (2+1)-dimensional gravity, the variable $\theta$ parametrising the conjugacy classes gives the mass of a particle in units of the Planck mass, or in our convention $\theta = \lambda m$. Geometrically, it gives the deficit (or surplus) angle of the conical spatial geometry surrounding the particle's worldline. 

The carrier spaces of UIRs of $\mathcal{D}(\mathrm{SU(1,1)})$ can, as with the Poincar\'{e} group, be given in terms of functions on $\mathrm{SU(1,1)}$ subject to an equivariance condition. The equivariance condition only depends upon the stabiliser subgroup, and in fact the carrier space  is undeformed. The action of $\mathcal{D}(\mathrm{SU(1,1)})$ on $V_{ms}$ is  a deformed version of  \eqref{pi ms rep} and is given by
\begin{equation}
\label{doubleequivariant}
(\Pie_{ms}(F)\psi )(v)=\int_{\mathrm{SU(1,1)}}F\left(g,g^{-1}ve^{-m\lambda s^0}v^{-1}g\right) \psi(g^{-1}v) \, d g,
\end{equation}
where $F \in \mathcal{D}(\mathrm{SU(1,1)})$ and $g,\:v \in \mathrm{SU}(1,1)$. To relate this formula to the  representation of the Poincar\'e group, it is useful to consider singular elements
\bee
F=\delta_h\otimes f, \quad h \in \mathrm{SU}(1,1), \; f\in C(\mathrm{SU}(1,1)).
\eee
Its action on $V_{ms}$  now  more closely resembles that of $(h,a)  \in \tilde{P}_3$ in the UIRs $\pie_{ms}$:
\bee
(\Pie_{ms}(\delta_h\otimes f)\psi ](v)=f(h^{-1}ve^{-m\lambda s^0}v^{-1}h)\psi(h^{-1}v), 
\eee
with $f$ generalising the function $\exp(i\langle a, \cdot\rangle)$.

In \cite{SchroersWilhelm}, local covariant fields are introduced for $\mathcal{D}(\mathrm{SU(1,1)})$ and  deformed momentum space (spin) constraints are derived. After Fourier transform these constraints are interpreted as deformed relativistic wave equations. In \cite{SchroersWilhelm} this is explicitly done for particles of spin $s=0, \frac{1}{2}$ and $1$.   We will not review these results here, but derive the analogues  for the anyonic case, where we need to consider the quantum double $\mathcal{D}(\widetilde{\mathrm{SU}}(1,1))$. The UIRs  are discussed in \cite{BaisScatt}. We only recall  the UIRs describing massive particles at this point, though we will need all  conjugacy classes when we consider the Fourier transform in the next section.

Massive particles are  described by UIRs with  conjugacy classes  obtained by exponentiating timelike generators of $\mathfrak{su}(1,1)$. Recalling from \eqref{expcover}  that mathematically positive rotations by an angle  $\mu$  have the form $\exp (-\mu s^0)=(\mu,0)$, we define
\begin{equation}
\label{TimelikeConClass}
E(\mu)=\Set{v(\mu, 0)v^{-1}|v \in \widetilde{\mathrm{SU}}(1,1),\;\mu \in (\RR \setminus 2\pi \ZZ)}.
\end{equation}
Elements in this class project to elliptic elements in $\mathrm{SU}(1,1)$ and our notation is chosen to reflect this. 

The interpretation of the unbounded   parameter $\mu$ in terms of 3d gravity is  something we will discuss in detail in Sect.~\ref{wavesect}.   For now we  again identify it with the mass in Planck units, i.e., 
\bee
\mu = \lambda m,
\eee
but note that  in the decomposition
\bee
\label{decompmu}
\mu = \mu_0 + 2\pi n, \quad n\in \ZZ, \quad \mu_0 \in  (0,2\pi),
\eee
only the  `fractional part'  $\mu_0$ has a classical  geometrical interpretation as a deficit angle. The integer parameter
\bee
n=\text{int}\left(   \frac {\mu}{2\pi}  \right)
\eee
 would affect gravitational Aharonov-Bohm scattering, as mentioned in the Introduction, but has no obvious classical meaning. In that sense, it is a `purely quantum' aspect of the particle. 

We will need a fairly detailed understanding of the  elliptic conjugacy classes $E(\mu)$ later in this paper, so we note that a generic element can be written  without loss of generality by choosing $v=(0,\beta)$, $\beta\in D$, so that 
\bee
(\omega,\gamma) =(0,\beta)(\mu,0)(0,-\beta),
\eee
and therefore, using \eqref{SU(1,1)tildeproduct}, 
\bee
\omega= \mu +\frac 1 i \ln \frac{1-|\beta|^2e^{i\mu_0}}{1-|\beta|^2 e^{-i\mu_0}}.
\eee
It is an elementary exercise to check that, if $\mu_0 < \pi$, then $\omega < \mu +\pi$, and if $\mu_0>\pi$ then 
$\omega < \mu +(2\pi-\mu_0)$.  It follows that 
\bee
\label{interparts}
\text{int} \left( \frac {\omega}{2\pi}\right)
=\text{int}\left(\frac {\mu}{2\pi}\right),
\eee
and this will be useful later.

The stabiliser subgroup of the representative element $(\mu,0)$ is given by
\begin{equation}
\tilde{N}^T=\Set{(\omega, 0) \in \widetilde{\mathrm{SU}}(1,1)} \simeq \RR.
\end{equation}
The stabiliser subgroup is the same as for UIRs  of $P^\infty_3$ describing massive particles, and its UIRs are labelled by a real-valued spin $s$. 

The carrier space of UIRs labelled by the mass parameter $m$ and spin $s$ is $V^A_{m s}$ as defined in \eqref{AnyonicEquivariance}.  Elements $F$  of the double $\mathcal{D}(\widetilde{\mathrm{SU}}(1,1))$ act on $\psi \in V^A_{ms}$ according to 
\begin{equation}
\label{coverdoubleequivariant}
(\Pie_{ms}(F)\psi )(v)=\int_{\widetilde{\mathrm{SU}}(1,1) }F\left(g,g^{-1}v(\lambda m, 0)v^{-1}g\right) \psi(g^{-1}v)d g,
\end{equation}

As  in the previous section, we now use the equivariant UIRs  to construct  a  covariant formulation.
\begin{definition}[Deformed Anyonic Covariant Field] The deformed anyonic covariant field  $\tilde{\phi}_{\pm}$ 
associated to an equivariant field $\psi\in V^A_{ms}$  is the map
\begin{align}
\label{gravancov}
\tilde{\phi}_{\pm} \colon E(\lambda m) \rightarrow \mathcal{H}_{l\pm},
\end{align}
where $\mathcal{H}_{l \pm}$ is the carrier space for the discrete series representations of $\widetilde{\mathrm{SU}}(1,1)$ \eqref{SL2 Coverings},
defined  via
\begin{align}
\tilde{\phi}_{\pm}(u)=\psi(v)D_{l\pm}(v)\Ket {0}_l.
\end{align}
Here $u,v\in \widetilde{\mathrm{SU}}(1,1)$, $v$ is chosen so that $u=v(\lambda m,0)v^{-1}$ and $s=l$ for $\tilde{\phi}_+$ and $s=-l$ for $\tilde{\phi}_-$.
\end{definition}

\begin{lemma}
The covariant fields $\tilde{\phi}_\pm$ are well-defined.
\end{lemma}
\begin{proof} One needs to check that the definition is independent of  the choice of $v$, i.e., that 
\begin{equation}
\psi(v) D_{l\pm}(v)\Ket {0_l}=\psi(v(\alpha,0)) D_{l\pm}(v(\alpha,0))\Ket {0}_l \quad \forall \alpha \in \RR.
\end{equation} This follows   by the calculation \eqref{invariance}  carried out  for the universal cover of the Poincar\'e group.
\end{proof}

The anyonic covariant field carries a unitary representation of $\mathcal{D}(\widetilde{\mathrm{SU}}(1,1))$
which we  denote $\Pic$:
\begin{equation}
\label{GravAnyonicCovAction}
(\Pic_{ms}(F) \tilde{\phi}_{\pm})(u)=\int_{\widetilde{\mathrm{SU}}(1,1)}  F(g,g^{-1}ug) D_{l\pm}(g)\tilde\phi_{\pm}(g^{-1}ug) \; d g.
\end{equation}
Without further condition, this representation is not irreducible.  We need  gravitised, 
 anyonic versions of  the  spin and mass constraints.

\begin{lemma} [Deformed Anyonic Spin Constraint]
\label{ribbonlemma}
The  anyonic fields \eqref{gravancov} satisfy the following spin constraint
\bee
\label{DefAnyCon}
\left( D_{l\pm}(u)-e^{i\lambda m s} \right) \tilde{\phi}_{\pm}(u)=0.
\eee
This constraint can be expressed in terms of the ribbon element as 
\bee
\Pic_{ms}(c)\tilde{\phi}_{\pm} = e^{i\lambda m s}  \tilde{\phi}_{\pm}.
\eee
\end{lemma}

\begin{proof} Let $u=v(\lambda m,0)v^{-1}$ where $v \in \widetilde{\mathrm{SU}}(1,1)$. Then, focusing on the positive series for simplicity, we compute
\begin{align}
D_{l+}(u)\tilde{\phi}_+(u)&= D_{l+}(u)\psi(v)D_{l\pm}(v)\Ket {0}_l \nonumber \\
&= \psi(v)D_{l+}(v)D_{l+}((\lambda m,0))\Ket {0}_l \nonumber \\
&= \psi(v)D_{l+}(v)e^{ il\lambda m}\Ket {0}_l \nonumber \\
&= e^{is\lambda ms}\tilde{\phi}_{+}(u) .
\end{align}
Using the action $\Pic_{ms}$ given in \eqref{GravAnyonicCovAction}, we also compute
\begin{align}
\label{ribbonact}
(\Pic_{ms}(c)\tilde{\phi}_+ )(u)&= \int_{\widetilde{\mathrm{SU}}(1,1)}  \delta_g(g^{-1}ug) D_{l+}(g)\tilde\phi_{+}(g^{-1}ug) \, dg \nonumber \\
&= D_{l+}(u)\tilde\phi_{+}(u),
\end{align}
thus confirming the second claim. The calculation for $\phi_-$  is entirely analogous.
\end{proof}

The mass constraint can be formulated using the projection map $\pi:\widetilde{\mathrm{SU}}(1,1) \rightarrow \mathrm{SU}(1,1)$ defined in \eqref{coverproject} and the relation \eqref{expoformula}. A necessary condition for an element  $u\in \widetilde{\mathrm{SU}}(1,1)$ to be in the conjugacy class $E(\lambda m) $  is 
\bee
\label{massconstraint}
\frac 12 \tr (\pi(u)) = \cos\left(\frac{\lambda m}{2}\right).
\eee 
This condition only  sees the fractional part in the decomposition \eqref{decompmu} of $\mu=\lambda m$. To ensure that $u \in 
E(\lambda m) $ we also need to impose the condition \eqref{interparts}. Writing $u=(\omega_u, \gamma_u)$ this reads 
\bee
\label{integerconstraint}
\text{int} \left( \frac {\omega_u}{2\pi}\right)
=\text{int}\left(\frac {\lambda m}{2\pi}\right). 
\eee
This is an analogue of the constraint relating the sign of the energy to that of the mass  in the representation theory of $P^\infty_3$. However,  it resolves an infinite  instead of a two-fold degeneracy and is not generally implied by the spin constraint \eqref{DefAnyCon}.

We thus arrive at the following carrier space for  the  anyonic covariant representation $\Pic$ 
of the double $\mathcal{D}(\widetilde{\mathrm{SU}}(1,1))$:
\begin{align}
\label{WGAms}
W_{ms}^{GA}= & \left\{ \tilde{\phi}_{\pm} \colon \widetilde{\mathrm{SU}}(1,1) \rightarrow \mathcal{H}_{l\pm} |  \left( D_{l\pm}(u) - e^{i\lambda ms}\right)\tilde{\phi}_{\pm}(u)=0, \right.
\nonumber \\
& \left.
\left(  \frac 12 \tr (\pi(u)) - \cos\left(\frac{\lambda m}{2}\right)\right) \tilde{\phi}_{\pm}(u)=0 ,
\left( \text{int} \left( \frac {\omega_u}{2\pi}\right)
-\text{int}\left(\frac {\lambda m}{2\pi}\right)\right) \tilde \phi 
(u)=0\right\},
\end{align}
where we choose the upper sign  for $s>0$ and the lower sign for $s<0$.

\begin{theorem}[Irreducibility of the carrier space $W^{GA}_{ms}$] The  covariant representation $\Pic_{ms}$  of  $\mathcal{D}(\widetilde{\mathrm{SU}}(1,1))$  on $W^{GA}_{ms}$  defined in  \eqref{GravAnyonicCovAction}     is unitarily equivalent to the  equivariant  representation $\Pie_{ms}$ on $V^{A}_{ms}$ defined in \eqref{coverdoubleequivariant}.  In particular, it is therefore irreducible.
\end{theorem}

\begin{proof}
We claim that the following maps  are intertwiners:
\begin{equation}
L_{\pm} \colon V^{A}_{ms} \rightarrow W^{GA}_{ms},\qquad (L_{\pm}(\psi))(u)= \psi(\omega, \gamma) D_{l\pm}(\omega, \gamma)\Ket {0}_l,
\end{equation}
where $(\omega,\gamma)\in\widetilde{\mathrm{SU}}(1,1)$ is chosen so that  $u=(\omega, \gamma)(\lambda m, 0) (\omega,\gamma)^{-1}$.
This follows the  steps in the proof of Theorem \ref{interparts}, but requires  replacing each statement for Lie algebras by the corresponding statement for groups.
Injectivity of the maps $L_\pm$ is immediate. To show  surjectivity,
we write the spin constraint of a  given state $\tilde \phi \in  W^{GA}_{ms}$, $s >0$ without loss of generality, as 
 \bee
  D_{l+}(\lambda m,0 )D_{l+}((\omega, \gamma)^{-1})\tilde \phi (u) = e^{i \lambda m  s} D_{l+}((\omega, \gamma)^{-1})\tilde \phi (u).
 \eee
Recalling that $l=s$, comparing with \eqref{InfAnyonRep+}, and also recalling that  $\ket{0}_l$ is, up to a factor, the unique solution of 
$ D_{l+}(\mu,0 )f = e^{i\mu l }  f$, we deduce the proportionality
\bee
D_{l+}((\omega, \gamma)^{-1})\tilde \phi (u) = \psi (\omega,\gamma) \ket{0}_l,
\eee
where the proportionality factor $\psi$ may depend on $(\omega,\gamma)$. Moreover, it must have the property
\bee
\psi (\omega+\alpha,\gamma) = e^{-i\alpha s} \psi(\omega,\gamma), 
\eee
to ensure independence of the choice of $(\omega,\gamma)$ for given $u$. Thus $\psi \in V^A_{ms}$.

The  intertwining property is equivalent to the  commutativity of the diagram
\bee
\begin{tikzcd}[row sep=large, column sep = large] 
V^A_{ms} \arrow{r}{\Pie_{ms}(F)} \arrow[swap]{d}{L_{\pm}} & V^A_{ms} \arrow{d}{L_{\pm}} \\
W^{GA}_{ms} \arrow{r}{\Pic_{ms} (F)} & W^{GA}_{ms}
\end{tikzcd},
\eee
for $F\in \mathcal{D}(\widetilde{\mathrm{SU}}(1,1))$.
This is a  straightforward calculation based upon the maps $L_{\pm}$, and the actions given in \eqref{doubleequivariant} and \eqref{GravAnyonicCovAction}.
\end{proof}

\subsection{Group Fourier transforms}

We now turn to the promised Fourier transform of  the covariant   UIRs of the Lorentz double. Our discussion here will  be less complete and rigorous than our treatment so far. In particular, we do not  survey  different approaches to Fourier transforms and differential calculus in the context of  Hopf algebras, but note that some relevant references are collected in \cite{SchroersWilhelm}. Instead, we only show how ideas first proposed by Rieffel in \cite{Rieffel} and recently pursued in the quantum gravity community under the heading of group Fourier transforms can be used to translate the algebraic mass and spin constraints in the definition \eqref{WGAms} into differential and difference equations.

 It is worth formulating the  problem we want to address
for a general Lie group $G$ with Lie algebra $\mathfrak{g}$. Concentrating for simplicity on complex-valued (rather than Hilbert space valued) functions, the standard Fourier transform \eqref{flatfourier}  is a map
\begin{equation}
\mathrm{L}^2(\mathfrak{g}) \rightarrow \mathrm{L}^2(\mathfrak{g}^*).
\end{equation}
However, in order to deal with the `gravitised' anyons, we require a Fourier transform
\begin{equation}
\mathrm{L}^2(G) \rightarrow \mathrm{L}_\star^2(\mathfrak{g}^*),
\end{equation}
where the $\star$ indicates the space has been equipped with a (generally non-commutative)  $\star$-product.

This is precisely the situation considered by Rieffel in \cite{Rieffel}, where he observed that, if the exponential map can be used to identify the Lie group with the Lie algebra, one can transfer the convolution product of functions on $G$ to functions on $\mathfrak{g}$ and then, by Fourier transform, to functions on $\mathfrak{g}^*$. This induces a non-commutative $\star$-product on functions on $\mathfrak{g}^*$ which is a strict deformation quantisation of the canonical Poisson structure on $\mathfrak{g}^*$. This works globally for nilpotent groups, but, as explained in \cite{Rieffel}, still makes sense, in an appropriate way, more generally.   For details we refer the reader to Rieffel's excellent exposition in the  paper \cite{Rieffel} which also contains comments on the  relation to other quantisation methods, such as Kirillov's coadjoint orbit method. 

 Ideas very similar to Rieffel's have, more recently and apparently independently, been considered by a number of authors in the context of quantum gravity \cite{FL,FM,Raasakka,GOR}. This work has resulted in a general framework called 
  group Fourier transforms. In developing  our Fourier transform  for gravitised anyons we essentially need to adapt and extend the ideas of Rieffel and the concept of a group Fourier transforms to $G=\widetilde{\mathrm{SU}}(1,1)$. 
  We    have found it convenient to use the terminology and notation used  in the discussion of group Fourier transforms, particularly in    \cite{Raasakka,GOR}, which we review briefly. 

The starting  point of the group  Fourier transform is the existence of non-commutative plane waves
\bee
E : G \times \mathfrak{g}^* \rightarrow \CC,
\eee
satisfying the following normalisation and  completeness relations
\begin{align}
\label{completeness} 
 E(e;x)&=1,  \nonumber \\
 E(u^{-1};x)&= \bar{E} ( u;x) = E(u;-x), \nonumber \\ 
\delta_e(u)&= \frac{1}{(2\pi)^d}\int_{\mathfrak{g}^*}E(u; x) \; dx, 
\end{align}
where $d$ is the dimension of $G$ and  $\delta_e$ is the  Dirac $\delta$-distribution  at the group identity element $e$   with respect to  the left Haar measure $dg$.  

Such non-commutative plane waves  induce a $\star$-product on a suitable set of functions on $\mathfrak{g}^*$ (to be specified below)  via the  group multiplication in $G$:  
\bee
\label{waveproduct}
E(u_1;x)\star E(u_2;x)=E(u_1u_2;x).
\eee
More precisely, given the non-commutative plane waves, one   defines 
\begin{align}
\mathcal{F}&: L^2(G) \rightarrow L^2_\star(\mathfrak{g}^*),  \nonumber \\
\phi(x)&=\mathcal{F}(\tilde \phi)(x)=\int_{G} E(u;x)\tilde{\phi}(u)\, du,
\end{align}
where  $ L^2_\star(\mathfrak{g}^*)$ is the image under $\mathcal{F}$ in  $ L^2(\mathfrak{g}^*)$, equipped with the $\star$-product defined by linear extension of \eqref{waveproduct} and with the  inner product imported from $ L^2(G)$. One checks that 
\bee
\langle \phi_1, \phi_2\rangle = \frac{1}{(2\pi)^d} \int_{\mathfrak{g}^\star} \bar{\phi}_1(x) \star \phi_2(x)\; dx = \int_G \bar{\tilde{\phi}}_1 \tilde{\phi}_2 \; dg .
\eee
By construction, this Fourier transform intertwines the convolution product on $L^2(G)$ with the star product on $ L^2_\star(\mathfrak{g}^*)$.

We also define a candidate for an inverse transform via
\begin{align}
\mathcal{F}^\star&: L^2_\star(\mathfrak{g}^*)\rightarrow  L^2(G), \nonumber \\ 
\tilde{\phi}(u)&=\mathcal{F}^\star(\phi)(u) =\frac{1}{(2\pi)^d} \int_{\mathfrak{g}^*}\bar{E}(u,x)\star \phi(x)\;dx ,
\end{align}
where  we emphasise the presence of the $\star$-product. 
It is easy to check  that completeness ensures that $ \mathcal{F}^\star \circ \mathcal{F}=\text{id}_{L^2(G)}$. However, $\mathcal{F} \circ \mathcal{F}^\star$ generally has a non-trivial kernel, see \cite{GOR}.

In \cite{GOR} it is shown that under certain assumptions, one can find  a coordinate map  $k:G\rightarrow \mathfrak{g}$ on $G$ and a function $\eta:G \rightarrow \CC$  so that,  up to a set of measure zero,  the plane waves take the form 
\begin{equation}
\label{PlaneWave}
E(u;x)= \eta(u) e^{i \langle x, k(u)\rangle}.
\end{equation}
Our task in the next section is to construct such non-commutative plane waves for $\widetilde{\mathrm{SU}}(1,1)$.

\subsection{Non-commutative waves  for $\widetilde{\mathrm{SU}}(1,1)$  and anyonic wave equations}
\label{wavesect}
Our proposal for a Fourier transform on  $\widetilde{\mathrm{SU}}(1,1)$ is based on the parametrisation of group elements summarised in the following proposition.

\begin{proposition}[Parametrisation of $\widetilde{\mathrm{SU}}(1,1)$]
Every element $(\omega, \gamma)\in \widetilde{\mathrm{SU}}(1,1)$ can be uniquely expressed   in terms of the $2\pi$-rotation  $\Omega$ \eqref{Omegadef} which generates the centre of $\widetilde{\mathrm{SU}}(1,1)$ and the exponential map via
\bee
\label{pnparametrisation}
(\omega, \gamma)=\Omega^n\widetilde{\exp}(p), \quad p = -\lambda \bp \cdot \bs \in \mathfrak{su}(1,1), \; n \in \ZZ, \; \lambda^2 \bp^2< (2\pi)^2, \; \text{and} \; \;p^0>0\; \text{if} \; \,\bp^2 >0.
\eee
\end{proposition}

Before we enter the proof, we should point out that, for elements in the elliptic conjugacy class $E(\mu)$ defined in \eqref{TimelikeConClass}, the integer $n$ introduced  in the Proposition is the same integer which appears in the decomposition \eqref{decompmu} of the rotation angle $\mu$. This follows since, for $(\omega,\gamma)\in E(\mu) $, 
\bee
(\omega,\gamma) = v(\mu,0)v^{-1}=\Omega^n \widetilde{\exp}(p)
\eee
implies 
\bee
 v^{-1}\widetilde{\exp}(p)  v = (\mu_0, 0), \quad 0<\mu_0 = \lambda |\bp| <2\pi,
\eee
so that $(\omega,\gamma)$ belongs to the conjugacy class with label  $\mu =\mu_0 + 2\pi n$. 

\begin{proof} To construct the claimed representation of a 
given $(\omega, \gamma)\in \widetilde{\mathrm{SU}}(1,1)$, we first compute the element 
$u=\pi(\omega,\gamma) \in \mathrm{SU}(1,1)$. Then, as reviewed in Sect.~2 and discussed in \cite{SchroersWilhelm}, the element $u$ or the element $-u$  is in the image of the exponential map in $\mathrm{SU}(1,1)$, i.e., there is a 
$p  \in \mathfrak{su}(1,1)$ and a choice of sign so that 
\bee
\label{uexp}
 u = \pm \exp(p).
\eee
However, then
\bee
\pi (\widetilde{\exp}(p)) = \pm u = \pm \pi(\omega,\gamma).
\eee
For the positive sign, this means  $(\omega,\gamma)\widetilde{\exp}(-p)$ is in the kernel of $\pi$, which is generated by   $\Omega^2$.  Thus $(\omega,\gamma) = \Omega^{2m}\widetilde{\exp}(p)$ for some $m\in \ZZ$  in this case. 
For the negative sign, we  recall that  $\pi(\Omega) =-\text{id}$, to deduce  $(\omega,\gamma) = \Omega^{2m+1}\widetilde{\exp}(p)$ for some $m\in \ZZ$ in that case. Thus  we obtain the  claimed decomposition \eqref{pnparametrisation}, with  even $n$  for the  positive sign in \eqref{uexp} and odd $n$ for the negative sign.

In order to establish uniqueness of the decomposition \eqref{pnparametrisation},   consider $p,p'\in \mathfrak{su}(1,1)$ both satisfying  the stated assumptions and $n,n'\in \ZZ$ so that   
\bee
\widetilde{\exp}(p)\Omega^n = \widetilde{\exp}(p')\Omega^{n'}.
\eee
We need to show $p=p'$ and $n=n'$. 
Projecting  into $\mathrm{SU}(1,1)$  we deduce
\bee
\label{test}
\exp(p)=\pm \exp(p').
\eee
In particular, $ \exp(p)$ and $\exp(p')$ must be of the same type, i.e, both must be either elliptic, parabolic or hyperbolic. 

We first consider the  case where either  $p$ or $p'$ vanishes. If one, say $p$, did then \eqref{test} would imply $ \exp(p')= \pm\text{id}$, but under the restriction on $p'$, this is only possible if the upper sign holds and $p'=0$, so that $n=n'$ follows.

Both parabolic and hyperbolic elements have the property that, if such an element is in the image of the exponential map, its negative is not.  For such elements we must therefore have the upper sign in \eqref{test}. Moreover, one checks from the expressions \eqref{expoformula} that parabolic and hyperbolic elements which are in the image of the exponential map have a unique logarithm, so that we conclude $p=p'$ and hence $n=n'$.  

Finally, elliptic elements differ from hyperbolic and parabolic elements in that both the element and its negative   are in the image of the exponential map, so that we must consider both signs in \eqref{test}. With either sign, that equation shows that $\exp(p)$ and $ \exp(p')$ commute  with each other. Then,  the explicit expression \eqref{expoformula} and 
\bee
\label{rangelimits}
0 <\lambda |\bp| <2\pi, \quad 0 <\lambda |\bp'| <2\pi,
\eee 
imply that $p$ and $p'$ must be multiples of each other. By the assumption that both lie in the forward light cone, we can deduce (recalling the sign conventions \eqref{signconventions}) that 
\bee
p= -\lambda |\bp|\hat p\cdot\bs , \quad p'=- \lambda|\bp'|\hat p\cdot\bs
\eee
Then \eqref{test} requires $\exp(p-p')=\pm\text{id}$ which is only possible if $ \lambda |\bp|$ and  $\lambda |\bp'|$ are equal or differ by $2\pi$ or $4\pi$. The last two possibilities are not compatible with \eqref{rangelimits}, and so we deduce $p=p'$ and $n=n'$ in this  case as well.
\end{proof}

\begin{figure}[h]
\begin{centering}
\includegraphics[width=4truecm,trim= 0 0.3cm 0 0]{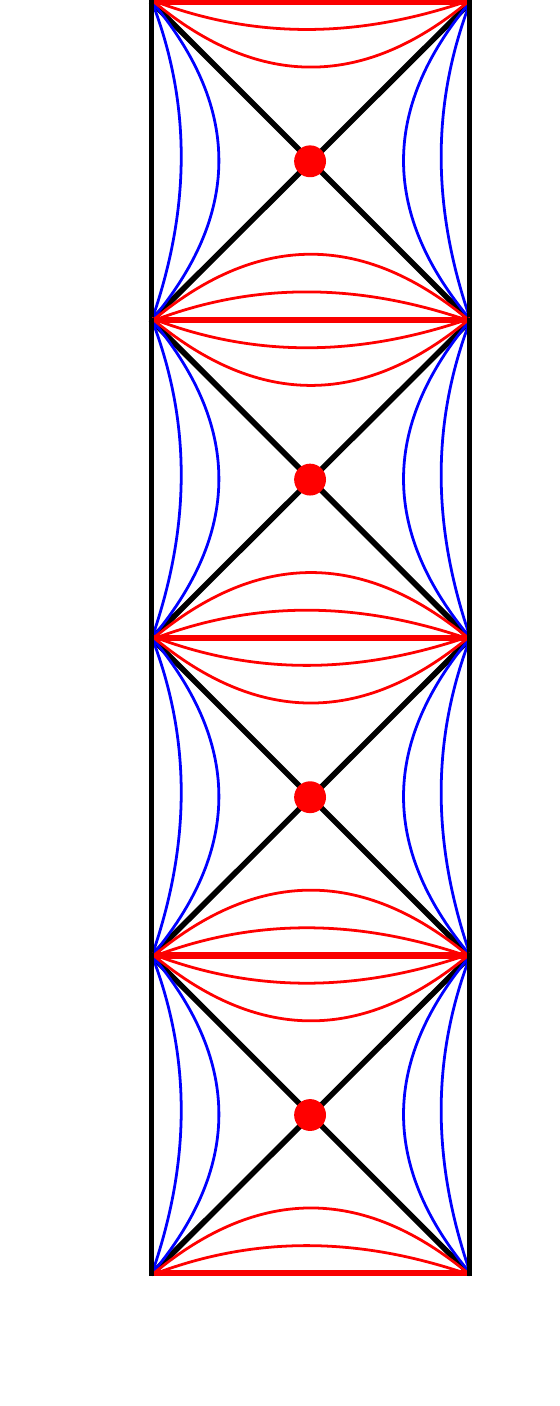}
\hspace{2cm}
\includegraphics[width=3.15truecm ]{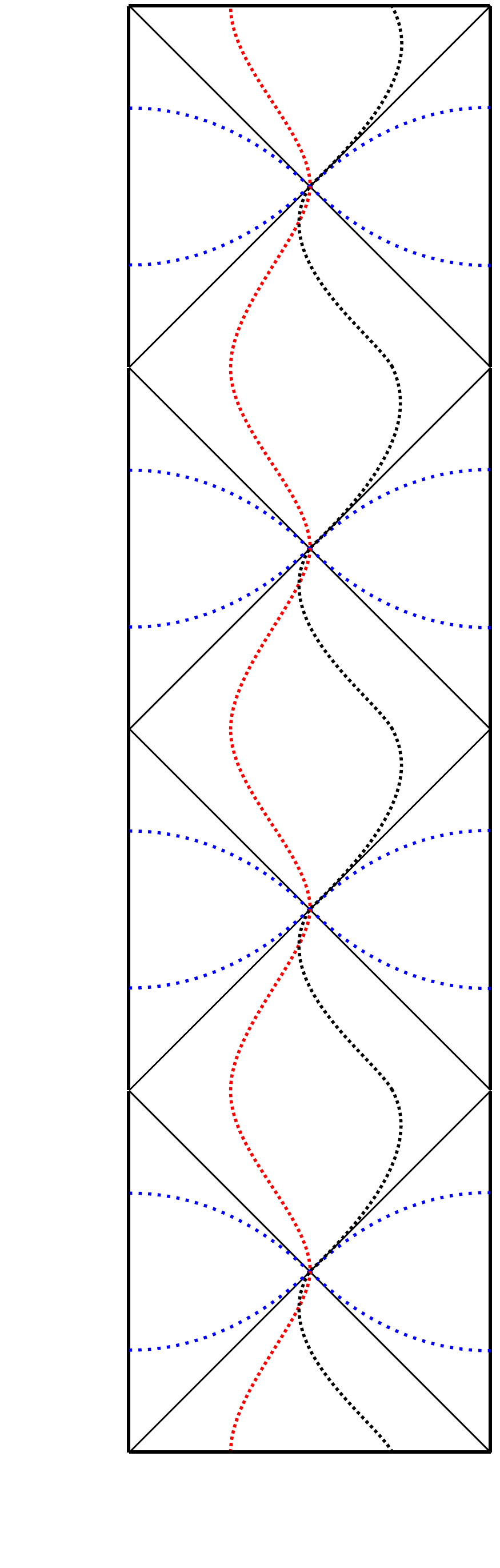}
\vspace{-0.4cm}
\caption{Conjugacy classes and exponential curves in  $\widetilde{\mathrm{SU}}(1,1)$. The diagram on the left shows selected  conjugacy classes   from the list \eqref{conjugacylist}, including  single element (red dots),  elliptic  (red),  parabolic (black) and hyperbolic conjugacy classes (blue).  The diagram on the 
right shows selected exponential curves \eqref{expocurves}  with initial tangent vector $p$ timelike (red), lightlike (black) and spacelike (blue).}
\label{Cylinder}
\end{centering}
\end{figure}

The decomposition \eqref{pnparametrisation} can be be visualised and illustrated by thinking of  $\widetilde{\mathrm{SU}}(1,1)$ as an  infinite cylinder, with  $\omega$ plotted along the vertical axis and  $\gamma$ parametrising the horizontal  slices. In Fig.~\ref{Cylinder} we show a vertical cross section of  this cylinder and display the conjugacy classes and the  exponential curves. 

A full list of conjugacy classes of  $\widetilde{\mathrm{SU}}(1,1)$ is given in  the appendix of \cite{BaisScatt}. There are  four kinds: single element conjugacy consisting of the elements $\Omega^n$, as well as elliptic, parabolic and hyperbolic conjugacy classes covering the corresponding classification for $\mathrm{SU}(1,1)$ discussed in Sect.~\ref{conventions}. In terms of the decomposition \eqref{pnparametrisation},   they can be described as follows: 
\begin{align}
\label{conjugacylist}
O^n&=\{\widetilde{\exp}(p)\Omega^n \in \widetilde{\mathrm{SU}}(1,1) | p=0\}, \nonumber \\
E^n_{\mu_0}& = \{ \widetilde{\exp}(p)\Omega^n \in \widetilde{\mathrm{SU}}(1,1) | p^0>0, \lambda |\bp| = \mu_0, 0< \mu_0   < 2\pi   \}, =E(\mu_0+2\pi n)\nonumber \\
P^n_+&= \{ \widetilde{\exp}(p)\Omega^n \in \widetilde{\mathrm{SU}}(1,1) | p^0>0,  \bp^2 =0   \}, \nonumber \\
P^n_-&= \{ \widetilde{\exp}(p)\Omega^n \in \widetilde{\mathrm{SU}}(1,1) | p^0<0,  \bp^2 =0   \}, \nonumber \\
H^n_\xi&= \{ \widetilde{\exp}(p)\Omega^n \in \widetilde{\mathrm{SU}}(1,1) | \bp^2= -\xi^ 2<0   \}.
\end{align}

The diagram on the right  in  Fig.~\ref{Cylinder}  shows  schematic sketches of exponential curves of the form
\bee
\label{expocurves}
\Gamma_{p,n}= \{\widetilde{\exp}(t p)\Omega^n| t\in [0,\infty)]\},
\eee
where $p$ is a fixed element of $\mathfrak{su}(1,1)$, and $n\in \ZZ$. In other words, these are images of the exponential map  with  chosen initial tangent vector $p$ translated by $\Omega^n$.  We stress that the cross section we are showing suppresses the three-dimensional nature of these curves. To illustrate this, we show three-dimensional plots of  some  exponential curves starting at the identity  in  Fig.~\ref{SampleGeodesics}. Note that the  spacelike and lightlike curves approach the boundary of the cylinder, but that the timelike curve winds  round the axis of the cylinder, carrying  out  a complete rotation when $\omega$ increases by $2\pi$

\begin{figure}
\begin{centering}
\includegraphics[width=11truecm]{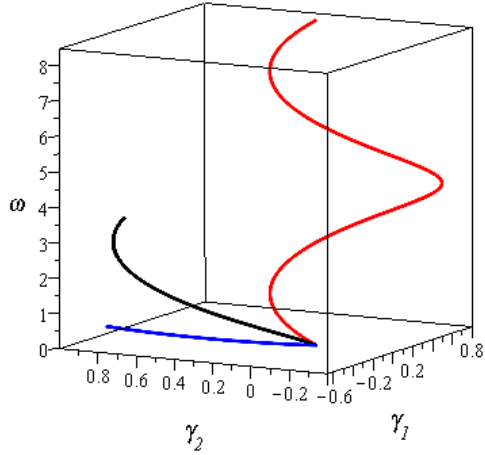}
\vspace{-6cm}
\caption{Exponential curves in $\widetilde{\mathrm{SU}}(1,1)$, obtained by exponentiating a spacelike (blue), lightlike (black) and timelike (red) tangent vector at the identity.}
\label{SampleGeodesics}
\end{centering}
\end{figure}

In order to compute the group Fourier transform we require  an expression for the Haar measure on $\widetilde{\mathrm{SU}}(1,1)$  in the coordinates \eqref{pnparametrisation}.
Using the abbreviations introduced in \eqref{expoformula}, we note that for  $
u= \exp(p)$, 
the left-invariant Maurer-Cartan form is 
\bee
u^{-1} \dd u = \lambda \dd |\bp| \hat{p}\cdot \bs  - 2 c  s\,  \dd \hat{p}\cdot \bs- 2 s^2 \hat{p}\times \dd \hat{p}\cdot{s},
\eee
where we use \eqref{pythagoras} and suppressed the argument $|\bp|$ of the functions $c$ and $s$ for readability. Thus 
\bee
[u^{-1} \dd u, u^{-1}\dd u] = 4s^2 \dd \hat{p}\times \dd \hat{p}\cdot \bs
- 4\lambda c s\, \dd  |\bp|\wedge  (\hat{p}\times \dd \hat{p})\cdot\bs + 4 \hat{p}^2 \lambda s^2 \dd |\bp| \wedge \dd\hat{p}\cdot\bs.
\eee
Thus, multiplying out and using again \eqref{pythagoras}, the Haar measure in exponential coordinates comes out  as 
\begin{align}
du &=\mathrm{tr}\left(u^{-1}\dd u \wedge u^{-1}\dd u \wedge u^{-1}\dd u\right) \nonumber \\
&= \frac{1}{2} \mathrm{tr}\left([u^{-1}\dd u,u^{-1}\dd u] \wedge u^{-1}\dd u\right)\nonumber \\
&=6 \lambda s^2\, \dd\hat{p}\times \dd\hat{p}\cdot \hat{p}\wedge \dd |\bp|.
\end{align}

Away from the set of measure zero where $\bp^2=0$, we have 
\bee 
d^3 \bp =\dd p^0\wedge \dd p^1 \wedge \dd p^2  = 
\frac 1 6 \dd\bp \times  \dd \bp \cdot \dd \bp 
= \frac 12 |\bp|^2  \dd\hat{p}\times \dd\hat{p}\cdot \hat{p}\wedge \dd  |\bp|, 
\eee
so that, again away from the set where $\bp^2=0$,
\bee
du=\rho(p)d^3\bp, \quad \text{with} \quad 
\rho(p)=  12   \lambda \frac{ s^2 ( |\bp|) }{ | \bp|^2 }.
\eee

Our parametrisation \eqref{pnparametrisation} of elements in $\widetilde{\mathrm{SU}}(1,1)$  requires both an element $p \in \mathfrak{su}(1,1)$ and an integer $n$. It is clear that a suitable non-commutative wave cannot depend only on  a dual variable $x\in \mathfrak{su}(1,1)^*$. It also requires an argument which is dual to the integer $n$. The most natural candidate is an angular coordinate, parametrising a circle $S^1$. The  necessity of a fourth and  circular dimension to describe the spacetime dual to $\widetilde{\mathrm{SU}}(1,1)$   is a surprise. We will introduce it and explore its consequences at this point, postponing a discussion to our final section. 

\begin{definition}We define non-commutative plane waves for $\widetilde{\mathrm{SU}}(1,1)$  as the maps
\begin{align}
\label{planewavedef}
E:  \widetilde{\mathrm{SU}}(1,1) \times (\mathfrak{su}(1,1)^* \times S^1) & \rightarrow  \CC, \nonumber \\
E(u;x,\varphi) =\frac{1}{\rho (p)} e^{i(\langle x,p\rangle + n\varphi)},
\end{align} 
where $p\in \mathfrak{su}(1,1)$ and $n\in \ZZ$ are the parameters determining $u$ via the decomposition \eqref{pnparametrisation}, and $\varphi\in[0,2\pi)$ is an angular coordinate on the circle $S^1$. 
\end{definition}

We need to check that the  non-commutative waves satisfy a suitable version of the completeness relation \eqref{completeness}. Expressing the $\delta$-function with respect to the left Haar measure  on $\widetilde{\mathrm{SU}}(1,1)$ in terms of the parameters $p$ and $n$ (see also B in \cite{GOR}),
\begin{equation}
\delta_e(g)=\frac{\delta_{n,0}}{\rho(p)}\delta^3(\bp),
\end{equation}
we confirm  the required condition:
\begin{align}
\frac{1}{(2\pi)^4}\int_{\mathfrak{su}(1,1)^* \times S^1} E(u;x,\varphi)\; d x d\varphi &=\frac{1}{2\pi} \int_{S^1}  e^{in\varphi }\,d \varphi  \; \frac{1}{(2\pi)^3}\int_{\RR ^3}  \,\frac{1}{\rho(p)}   e^{i\bp\cdot \bx}\, d^3 \bx \nonumber\\
&=\delta_{n,0}\, \frac{1}{\rho(p)}\delta^3(\bp) \nonumber \\
&=\delta_{e}(u).
\end{align}

Before we use our non-commutative waves to carry out the   group Fourier transform,  we make some observations and comments. 
With the terminology explained after \eqref{flatfourier}, the  non-commutative waves 
\bee
e^{i(\langle x,p\rangle + n\varphi)}= e^{i(\bx \cdot \bp + n\varphi)}
\eee
look like  standard plane waves on the product of  Minkowski space with a circle. However, the momentum $\bp$ is constrained by the conditions in \eqref{pnparametrisation}, so that timelike momenta have an invariant mass  which is  bounded from above by 
\bee
\label{Planckmass}
m_p= \frac{2\pi}{\lambda} =\frac{1}{4G},
\eee
and are always in the forward lightcone.  The existence of  the Lorentz-invariant Planck mass  $m_p$, and 
 hence also of an invariant Planck length, in the Lorentz double is one of its important features. It means in particular that it provides an example of a `doubly special theory of relativity' in 2+1 dimensions which neither deforms nor breaks Lorentz symmetry,  see \cite{Lessons3DGrav} and our Summary and Outlook  for further comments on this point.

It is natural to interpret the integer $n$ in the spirit of particle physics as a label for different kinds of particles in the theory. Timelike momenta $p$ with $n=-1$  may then be viewed as describing antiparticles. Lightlike and spacelike momenta for $n=0$ have the usual interpretation as momenta of massless or (hypothetical) tachyonic particles. The other values of $n$ describe additional types of massive, massless and tachyonic particles. Their existence is required by the fusion rules obeyed by the plane waves, which follow from the star product
\bee
\label{starwaves}
e^{i(\langle x,p_1\rangle + n_1\varphi)}\star e^{i(\langle x,p_2\rangle + n_2\varphi)}  =e^{i(\langle x,p(u_1u_2)\rangle + n(u_1u_2)\varphi)}.
\eee

The general features of this fusion rule can be read off from the picture of the conjugacy classes of $\widetilde{\mathrm{SU}}(1,1)$ on the left in Fig.~\ref{Cylinder}. When multiplying plane waves for particles of types  $n_1$ and  $ n_2$, the particle type of the combined system is determined by the product, in  $\widetilde{\mathrm{SU}}(1,1)$, of the group-valued momenta $u_1$ and $u_2$. This is a generalisation of the well-know Gott-pair  in 2+1 gravity, where two ordinary particles ($n=0$) with high relative speed can combine into a particle with tachyonic momentum (and, in our terminology, of type $n=1$).

Thus we think of the plane waves for  $\widetilde{\mathrm{SU}}(1,1)$ as describing kinematic states of particles in a theory  with an invariant mass scale $m_p$  and with infinitely many different types of particles which combine according to the fusion rules encoded in the star product.

The Fourier transform of a covariant field 
$
\tilde \phi: \widetilde{\mathrm{SU}} (1,1) \rightarrow \mathcal{H}_{l\pm}$
is given by
\begin{align}
\label{spacetimefield}
\phi(x,\varphi)   =\int  E(u;x,\varphi)\tilde\phi(u) \, du   
  =\sum_{n\in \ZZ} \int_{B_+} e^{i(\langle x,p\rangle + n\varphi)}\tilde \phi (u) \, d^3\bp,
\end{align}
where 
\bee
B_+= \{p =-\lambda \bp \cdot \bs \in \mathfrak{su}(1,1)| \lambda^2 \bp^2 <(2\pi)^2, p^0>0 \; \; \text{if} \; \bp^2 >0 \} 
\eee
is the region in momentum space required in the parametrisation \eqref{pnparametrisation}.  

Since the field $\phi$ takes values in the Hilbert space $\mathcal{H}_{l\pm}$, the extension of the $\star$-product \eqref{starwaves} to products of such fields  requires a careful tensor product decomposition. We will not pursue this here, but note that the  decomposition of  tensor products in the equivariant formulation of the UIRs for  quantum doubles of compact Lie groups was studied in detail in \cite{KBM}. A full study of the $\star$-product for covariant fields $\phi$ will require an extension to non-compact groups  and an adaptation  of the results of that paper to our covariant formulation.

Here, we focus on single fields and  apply the group   Fourier transform to   the deformed anyonic spin constraint  \eqref{DefAnyCon} and the mass  constraints \eqref{massconstraint} and \eqref{integerconstraint}. 
Expanding $u\in \mathrm{SU}(1,1)$ as in \eqref{pnparametrisation}, we first note that 
\bee
 D_{l\pm}(u)=  D_{l\pm}(\Omega^n\widetilde{\exp}(p)) = e^{2\pi i  ns } D_{l\pm}(\widetilde{\exp}(p)) = e^{2\pi i ns } e^{d_{l\pm}(p)}.
\eee
But then, with $\lambda m =\mu$ and the decomposition  \eqref{decompmu}, we also have 
\bee
e^{i\lambda m s} = e^{i\mu_0 s} e^{2\pi ins}. 
\eee
Hence, with $p=-\lambda s^ap_a$,  the momentum space spin constraint \eqref{DefAnyCon} is equivalent to 
\bee
\label{DefAnyConn}
\left(e^{-\lambda d_{l\pm}(s^a)p_a}-e^{i\mu_0 s} \right) \tilde{\phi}_{\pm}(u)=0.
\eee
The mass constraints \eqref{massconstraint} and \eqref{integerconstraint} also take a simple form in the parametrisation \eqref{pnparametrisation}. The former only see the fractional  part $\mu_0\in (0,2\pi)$ of  $\lambda m$, and is equivalent to
\bee
\label{newmassconstraint}
 \bp^2 =\frac{\mu_0^2}{\lambda^2}.
\eee
The condition   \eqref{integerconstraint}  simply fixes the integer $n$ in the decomposition \eqref{pnparametrisation}.

Applying the Fourier transform  \eqref{spacetimefield} turns the algebraic momentum space  constraints into differential equations.
 The mass constraint \eqref{newmassconstraint} becomes the  Klein-Gordon equation  for the fractional part of the mass:
 \bee
\left( (\partial_0^2-\partial_1^2-\partial_2^2)+\frac{\mu_0^2}{\lambda^2}\right)\phi(x,\varphi)=0.
 \eee
 The integer constraint \eqref{integerconstraint}  fixes the integer part of the mass
 via the differential condition on the angular dependence of $\phi$:
 \bee
 -i\frac{\partial}{\partial \varphi} \phi (x,\varphi)= n\phi(x,\varphi).
 \eee
Finally, the spin constraint \eqref{DefAnyConn}  becomes
\begin{align}
\label{exponentialspace}
\left( e^{ i\lambda d_{l\pm}(s^a)\partial_a}-e^{i\mu_0 s} \right) \phi_{\pm}(x,\varphi)=0.
\end{align}
This equation involves an  exponential of  the differential operators \eqref{anywave2+}and \eqref{anywave2-},  combining  spacetime derivatives  with complex derivatives in the hyperbolic disk. 
This is the anyonic generalisation of the exponential Dirac operator $e^{-\frac{\lambda}{2}\gamma^a\partial_a}\phi (x)$ that was obtained in \cite{SchroersWilhelm} for the massive spin $\frac{1}{2}$ particle. 

Similar exponential operators have been considered in a more general context in \cite{Atiyah}, where it was stressed that they are essentially finite difference operators.  The appearance of difference-differential equations was first  mentioned in relation to (2+1) gravity in \cite{Matschullparticle}. 

It is clear that  further work is required to make sense of the equation \eqref{exponentialspace}. Stripped down to its simplest elements (by reducing all dimensions to one), it is an equation of the form
\bee
e^{\lambda \frac{\dd}{\dd x} }\phi(x)=  e^{\lambda k}\phi(x),
\eee
or, assuming analyticity, 
\bee
\phi(x+\lambda)= e^{\lambda k} \phi(x).
\eee
For analytic functions, this is equivalent to the infinitesimal version
\bee
\frac{\dd\phi}{\dd x} = k \phi(x),
\eee
obtained by differentiating with respect to $\lambda$.  This simple example suggests that the anyonic constraint \eqref{anywave1} and the gravitational anyonic constrained \eqref{exponentialspace} may, suitably defined, be infinitesimal and finite versions of the same condition. However, careful analysis is required to clarify the definition of \eqref{exponentialspace} and its relation to \eqref{anywave1}.

\section{Summary and Outlook }

This paper was motivated by the observation that,  in the context of  2+1 dimensional quantum gravity,   the spin quantisation \eqref{SpinQuant} forces one to  consider the universal cover of the  Lorentz  group and that, in order to preserve the duality between momentum space and Lorentz transformations in the quantum double, it  is natural to take  the universal cover in momentum space, too.

We showed how  the representation theory of the quantum double of the universal cover 
$\widetilde{\mathrm{SU}}(1,1)$ can be cast in a Lorentz-covariant form,  and can be  Fourier transformed. In this process, the universal covering of the Lorentz group  necessitates   the use of infinite-component fields, but the universal covering of momentum space has more interesting and far-reaching consequences.

The first of these, exhibited  both in the  decomposition \eqref{decompmu} and in the parametrisation \eqref{pnparametrisation}, is the extension of the range of the allowed mass. The fractional part $\mu_0$ of $\mu=8\pi Gm$ is the conventional mass  of a particle, which manifests itself in classical (2+1)-dimensional gravity as a conical deficit angle in the spacetime surrounding  the particle. The integer label $n$ in the decomposition $\mu= \mu_0 +2\pi n$  appears  to be a purely quantum observable with no classical analogue. 
It manifests itself, for example, in the  Aharonov-Bohm scattering cross section of two massive particles, as discussed in \cite{BaisScatt}.  It is a rather striking illustration of the concept of  `quantum modular observables' as introduced in \cite{APP},  extensively discussed in the textbook \cite{AR} and recently  applied to the notion of spacetime in \cite{FLM}. 

We have chosen to interpret $n$ as  a label of different types  of particles or matter in (2+1)-dimensional quantum gravity. These particles can be converted into each other during interactions, according to fusion rules determined by the group product in  $\widetilde{\mathrm{SU}}(1,1)$ and the decomposition of factors and products according \eqref{pnparametrisation}.
  
The second and surprising consequence of the universal  covering of momentum space is the appearance of an additional and compact dimension on the dual side, in spacetime. This is  needed to define the group Fourier transform, and allows for a simple expression of the constraint determining the particle type $n$  as a differential condition.

Our results raise  a number of  questions and suggest   avenues for future research.  As  discussed at the end of the previous section, the exponentiated differential operators which generically appear as group Fourier transforms of  the spin constraint  should be studied using rigorous analysis. One expects these to be natural operators, possibly best defined as difference operators, not least because they are, by Lemma~\ref{ribbonlemma},  essentially Fourier transforms of the ribbon element of the quantum double.

It seems clear that our  Hilbert-space valued fields  $\phi$ on Minkowski space equipped with  a $\star$-product fit rather naturally into the framework of braided quantum field theories, defined in \cite{Oeckl} and studied, in a  Euclidean setting, in \cite{FL,SS}. Our paper is  only concerned with a single particle and we only looked at  simple examples of  fusion rules for two spinless particles.  However, braided quantum field theory naturally  provides the language for discussing the  gravitational interactions of several gravitational anyons  in a spacetime setting. This  provides an alternative viewpoint to existing  momentum space discussions, with the non-commutative $\star$-product and the universal $R$-matrix of the quantum double encoding the quantum-gravitational interactions.

It is worth stressing that the $\star$-product on Minkowski space considered here preserves Lorentz covariance, and that any braided quantum field theory constructed from it would similarly be Lorentz covariant. This follows  essentially from the Lorentz invariance of the mass scale  \eqref{Planckmass} and the associated Planck length scale \cite{NonCommutSch}. It reflects the important fact that the Lorentz double deforms Poincar\'e symmetry by introducing a  mass scale while preserving Lorentz symmetry, which is a challenge for any theory of quantum gravity. 

Finally, it would be interesting to repeat the analysis of this paper with the  inclusion of a cosmological constant. This leads to a $q$-deformation of the quantum double of  $\widetilde{\mathrm{SU}}(1,1)$ 
\cite{NonCommutSch}, and there should similarly be a $q$-deformation of the spacetime picture of the representations.  Some remarks on how this might work are made in \cite{SemiDual}, but none of the details have been worked out. In the Lorentzian context, a positive cosmological constant will  lead to real deformation parameter $q$ while a negative cosmological constant will lead to $q$ lying on the unit circle \cite{NonCommutSch}. It would clearly be interesting to understand how this change in $q$ captures the radically different physics of the two regimes.

\vspace{0.5cm}
\noindent {\bf Acknowledgements} \; SI  acknowledges support through an EPSRC doctoral training grant. Several of the results in this paper were reported by BJS at the Workshop on Quantum Groups in Quantum Gravity,  University of Waterloo 2016. BJS thanks the organisers for the invitation, and acknowledges discussions with participants at the workshop. We thank Peter Horvathy for drawing our attention to Plyushchay's work on  anyonic wave equations and to their joint work on noncommutative waves.

\end{document}